\documentclass[%
 amsmath,amssymb,
superscriptaddress
]{revtex4-1}

\usepackage{graphicx}
\usepackage{dcolumn}
\usepackage{bm}
\usepackage{textcomp} 

\usepackage{physics}
\newcommand{\xb}{{\boldsymbol x}}
\newcommand{\yb}{{\boldsymbol y}}
\newcommand{\yp}{y^\prime}

\newcommand{\alb}{{\boldsymbol \alpha}}
\newcommand{\beb}{{\boldsymbol \beta}}
\newcommand{\kp}{{k^\prime}}
\newcommand{\lp}{{l^\prime}}
\newcommand{\anglink}{{\left \langle i \right \rangle_k^n}} 
\newcommand{\angljnpk}{{\left \langle j \right \rangle_k^{n^\prime}}} 
 
\newcommand{\angljnkp}{{\left \langle j \right \rangle_\kp^n}} 
\newcommand{\anglikkp}{{\left \langle i \right \rangle_\kp^k}} 
 
\newcommand{\anglikl}{{\left \langle i \right \rangle_\lambda^k}}

\newcommand{\matr}[1]{\mathbf{#1}}
\newcommand{\pt}{\tilde{p}}
\newcommand{\dkl}{\mathcal{D}_{\mathcal{K}\mathcal{L}}}
\newcommand{\oset}{\varnothing}
\newcommand{\vnu}{{\boldsymbol \nu}}
\newcommand{\vmu}{{\boldsymbol \mu}}
\newcommand{\fknu}{F_k(\{\nu\})}

\newcommand{\tp}{t^\prime}
\newcommand{\ovnu}{\overline{\nu}}
\newcommand{\ovf}{\overline{F}}

\newcommand{\bb}{{\boldsymbol b}}
\newcommand{\mb}{{\boldsymbol m}}
\newcommand{\deltaHI}{\hat{\delta}}

\makeatletter
\newcommand*{\rightharpoonupfill@}{%
  \arrowfill@\relbar\relbar\rightharpoonup
}
\newcommand*{\leftharpoondownfill@}{%
  \arrowfill@\leftharpoondown\relbar\relbar
}
\newcommand{\xrightleftharpoons}[2][]{%
  \ensuremath{%
    \mathrel{%
      \settoheight{\dimen@}{\raise 2pt\hbox{$\rightharpoonup$}}%
      \setlength{\dimen@}{-\dimen@}%
      \edef\CA@temp{\the\dimen@}%
      \settoheight\dimen@{$\rightleftharpoons$}%
      \addtolength{\dimen@}{\CA@temp}%
      \raisebox{\dimen@}{%
        \rlap{%
          \raisebox{2pt}{%
            $%
            \ext@arrow 0359\rightharpoonupfill@{\hphantom{#1}}{#2}%
            $%
          }%
        }%
        \hbox{%
          $%
          \ext@arrow 3095\leftharpoondownfill@{#1}{\hphantom{#2}}%
          $%
        }%
      }%
    }%
  }%
}
\makeatother

\usepackage{algorithm}
\usepackage{algorithmicx}
\PassOptionsToPackage{noend}{algpseudocode}
\usepackage{algpseudocode}

\makeatletter
\let\OldStatex\Statex
\renewcommand{\Statex}[1][3]{%
  \setlength\@tempdima{\algorithmicindent}%
  \OldStatex\hskip\dimexpr#1\@tempdima\relax}
\makeatother

\algnewcommand{\LineComment}[1]{\State \(\triangleright\) #1}

\usepackage{amsthm}
\newtheorem{prop}{Proposition}

\begin{document}

\title{Learning Dynamic Boltzmann Distributions as Reduced Models of Spatial Chemical Kinetics}

\author{Oliver K. Ernst}
\email{oernst@ucsd.edu.}
\affiliation{
Department of Physics, University of California at San Diego, La Jolla, California
}
\author{Thomas Bartol}
\email{bartol@salk.edu.}
\author{Terrence Sejnowski}
\email{terry@salk.edu.}
\affiliation{
Salk Institute for Biological Studies, La Jolla, California
}
\author{Eric Mjolsness}%
\email{emj@uci.edu.}
\affiliation{
Departments of Computer Science and Mathematics, and Institute for Genomics and Bioinformatics, University of California at Irvine, Irvine, California
}

\date{\today}


\begin{abstract}

Finding reduced models of spatially-distributed chemical reaction networks requires an estimation of which effective dynamics are relevant.
We propose a machine learning approach to this coarse graining problem, where a maximum entropy approximation is constructed that evolves slowly in time. 
The dynamical model governing the approximation is expressed as a functional, allowing a general treatment of spatial interactions. 
In contrast to typical machine learning approaches which estimate the interaction parameters of a graphical model, we derive Boltzmann-machine like learning algorithms to estimate directly the functionals dictating the time evolution of these parameters.
By incorporating analytic solutions from simple reaction motifs, an efficient simulation method is demonstrated for systems ranging from toy problems to basic biologically relevant networks.
The broadly applicable nature of our approach to learning spatial dynamics suggests promising applications to multiscale methods for spatial networks, as well as to further problems in machine learning.

\end{abstract}

\maketitle


\section{\label{sec:1}Introduction}


\subsection{\label{sec:1.A}Model Reduction of Statistical Many-Body Systems}

Master equations are broadly applicable to stochastic systems in biology. 
For reaction-diffusion systems, the solution to the chemical master equation (CME) fully characterizes the probability distribution over system states and all observables at all times~\cite{gardiner_1976}. 
However, solving the CME for relevant moments is challenging when the interactions of two or more reagents lead to non-linear differential equation systems for the moments, and even more challenging when considering spatially distributed systems. 
A wealth of analytical and numerical approaches have been developed in pursuit of approximate solutions, each of which is optimally suited for a distinct dynamical regime~\cite{gillespie_2013}. 
For example, at the low and spatially heterogeneous concentrations of molecules present in dendritic spines in synapses, particle-based methods may best describe the highly stochastic signaling activity~\cite{bartol_2015,mcell_1,mcell_2}. 
In the larger volumes such as the dendritic shaft, simpler geometries and higher concentrations allow more efficient partial differential equation (PDE) methods. 
It remains an open problem to develop a modeling framework that is able to flexibly transition across different dynamical regimes, or to describe their coexistence in the same spatial domain.

One key problem is that the number of states appearing in the CME increases exponentially with the number of underlying random variables describing the system. 
This system state space explosion poses a computational challenge for Monte Carlo algorithms such as the popular Gillespie stochastic simulation algorithm (SSA)~\cite{gillespie_1977}, requiring the sampling of a sufficiently large number of trajectories to estimate observables.

The direct estimation of observables also poses a challenge. 
Generally, many-body systems result in a hierarchy of moments (analogous to a BBGKY hierarchy~\cite{bogoliubov_1946,kirkwood_1946,kirkwood_1947}), where the differential equation for any given moment depends on the current value of higher order ones. 
This requires the use of a moment closure technique (see Ref.~\onlinecite{johnson_2015} for review), but a poor choice here can unduly restrict modeling of the rich correlation structures available.

Machine learning approaches present an opportune setting for addressing these problems, because a central goal of these approaches is the estimation of the structures underlying complex correlations. 
For example, machine learning has recently been proposed to approximate quantum systems~\cite{carleo_2017}. 
Previous applications to chemical systems include the predictions of molecular reactions~\cite{kayala_2011} and synaptic activity~\cite{montes_2013}. 
However, to our knowledge no general formulation for learning chemical dynamics exists that incorporates the complex spatial interactions central to many problems in biology. 
In this work, we present such a framework and derive algorithms for simulating reaction-diffusion systems in continuous space. 
This has promising general implications for both moment closure and model reduction of the CME, and of other generally spatially distributed systems.


\subsection{\label{sec:1.B}Inferring Markov Random Fields for Reduced Dynamics}

Previous work has shown the applicability of machine learning to model reduction. 
In particular, the ``Graph Constrained Correlation Dynamics" (GCCD) framework~\cite{johnson_2015} uses a Markov Random Field (MRF) of plausible state variables and interactions as input, incorporating human expertise into the model reduction process. 
The probability distribution associated with this MRF is written in a form that separates the time evolution $\mu(t)$ from the graph structure $V_\alpha(s)$:
\begin{equation}
\pt(s, t; \{ \mu \}) = \frac{1}{\mathcal{Z}(\mu(t))} \exp [ - \sum_\alpha \mu_\alpha(t) V_\alpha(s) ] .
\label{eq:gccd}
\end{equation}
This time-evolving mean field model can be learned separately at each time-point using the well known Boltzmann Machine (BM) learning algorithm~\cite{ackley_1985}, and then approximated with its own dynamics for $\mu(t)$. 
For a suitably chosen MRF, the resulting temporal dynamics of $\mu(t)$ result in a large degree of model reduction, as reported in Ref.~\onlinecite{johnson_2015}. 

In the next section, we formalize these ideas and extend to the spatial domain a general variational problem for estimating the dynamical model dictating the time evolution of a Boltzmann distribution. 
The importance of spatial networks has been widely studied~\cite{durrett_1994,barthelemy_2011}, with continued interest in mean-field methods, such as for evolving networks~\cite{king_2017}. 
Our formulation using functionals presents a flexible framework for capturing the dynamics of any desired spatial correlations in the system. 
This leads to algorithms closely related to a BM, with a modified learning rule for directly estimating the functionals dictating the time evolution of the mean-field model. 
We anticipate that such an approach is broadly applicable to other spatially organized networks, and will have further practical applications in machine learning.


\section{\label{sec:2}Learning Algorithms for Model Reduction}

This section is organized as follows. 
Section~\ref{sec:2.A} reviews recent developments that have enabled the derivation of stochastic simulation algorithms for chemical kinetics in the Doi-Peliti formalism, and introduces further extensions for describing spatial dynamics. 
Section~\ref{sec:2.B} introduces time-evolving Boltzmann distributions as reduced models, and Section~\ref{sec:2.C} sets up the associated variational problem for their dynamics. 
For the case of well-mixed systems, this is solvable by an algorithmic approach, as shown in Section~\ref{sec:2.D}. 
Section~\ref{sec:2.E} treats some analytic solutions for simple systems, used to guess parameterizations needed to algorithmically solve the general spatially heterogeneous case in Section~\ref{sec:2.F}.


\subsection{\label{sec:2.A}Field Theoretic Approaches to Deriving Stochastic Simulation Algorithms}

An equivalent description parallel to the CME is the  quantum-field-theoretic Doi-Peliti~\cite{doi_1976_1,doi_1976_2,peliti_1985} operator algebra formalism (see e.g. Ref.~\onlinecite{mattis_1998} for review). 
Extensions to this formalism have developed it as a natural framework for deriving stochastic simulation algorithms of chemical kinetics. 
In particular:
\begin{enumerate}
\item The introduction of parameterized objects has generalized bare molecules to allow the description of macromolecular complexes such as phosphorylation states~\cite{mjolsness_2006,mjolsness_2013_1}, or other structures with size, type, and other internal parameters that affect their dynamics.
\item Dynamically graph-linked collections of objects~\cite{mjolsness_2010} have allowed collections of inter-related objects and extended objects to associate and dissociate according to specified rules and propensities.
\item Differential operators have been introduced that express differential equations and stochastic differential equations~\cite{mjolsness_2006,mjolsness_2013_1}, as in the Lie Series~\cite{singer_1990}.
\end{enumerate}

These innovations lead naturally to the rederivation of the popular Gillespie SSA from the CME, and from there to extensions to the parametric, graph-matching, and mixed-dynamics cases~\cite{mjolsness_2013_1}.
Here, we consider further extensions to this formalism to develop {\it model reduction} techniques for spatial reaction-diffusion systems.

The raising and lowering operators, $\hat{a}$ and $a$, create and destroy identical particles of a single species. 
For states consisting of a single species distributed on a discrete lattice $\ket{\{ n \}}$, where $\{ n \}$ describes the occupancy of each lattice site, the action of these operators on the $i$-th lattice site (in some ordering) is
\begin{equation}
\begin{split}
\hat{a}_i \ket{\{ n \}} &= \ket{\{ \dots, n_{i-1}, n_i + 1, n_{i+1}, \dots \}} , \\
a_i \ket{\{ n \}} &= n_i \ket{\{ \dots, n_{i-1}, n_i - 1, n_{i+1}, \dots \}} .
\end{split}
\end{equation}
Further, they satisfy the Heisenberg algebra commutation relationship $[a_i, \hat{a}_j] = \delta_{ij}$ where $\delta_{ij}$ is the Kronecker delta function.
We note that these are different from the ladder operators in quantum mechanics in that $a_i$ is not conjugate to $\hat{a}_i$. 
However, in the present context they are of key importance as they capture mass action chemical kinetics.

These operators admit an equivalent generating function representation:
\begin{equation}
\ket{ \{ n \} } \rightarrow \prod_i z_i^{n_i} ,
\end{equation}
where the product runs over all spatial lattice sites. 
Then the operators may be represented as:
\begin{equation}
\begin{split}
\hat{a}_i &\rightarrow z_i , \\
a_i &\rightarrow \frac{\partial}{\partial z_i} .
\end{split}
\end{equation}

In the spatially continuous case, let the state of the system be denoted by $\ket{n,\alb,\xb}$, consisting of $n$ particles at locations $\xb$, consisting of $n$ positions in 3D space, with species labels $\alb$, also of length $n$. 
The equivalent generating function representation is:
\begin{equation}
\ket{n,\alb,\xb} \rightarrow \prod_{i=1}^n z(\alpha_i,x_i).
\end{equation}
The raising and lowering operators are now:
\begin{equation}
\begin{split}
\hat{a}(\alpha,x) &\rightarrow z (\alpha,x), \\
a(\alpha,x) &\rightarrow \frac{\delta}{\delta z(\alpha,x)},
\end{split}
\end{equation}
where, switching from the discrete to the continuous case, partial derivatives for the annihilation operator turn into functional derivatives. 
Importantly, the CME
\begin{align}
\dot{p}(n,\alb,\xb,t) = \matr{W} p(n,\alb,\xb,t)
\label{eq:cme}
\end{align}
can still be written in an equivalent form where the time-evolution operator $\matr{W}$ is polynomial in the ladder operators,
encoding the set of reactions and rates. 
We make use of these extensions in the following sections where analytic forms for differential equations of moments are required.


\subsection{\label{sec:2.B}Reduced States in a Dynamic Boltzmann Distribution}

Let $\ket{n,\alb,\xb,t}$ denote the true state of the system at time $t$. 
In the spirit of a MRF, construct states in a coarse-scale model:
\begin{widetext}
\begin{equation}
\ket{ \{ \nu_k \}_{k=1}^K, t } = \frac{1}{\mathcal{Z} \left [ \{ \nu_k \}_{k=1}^K \right ]} \sum_{n=0}^\infty \sum_\alb \int d\xb \; \exp \left [ 
- \sum_{k=1}^K \sum_\anglink \nu_k (\alb_\anglink, \xb_\anglink, t) \right ] \ket{n, \alb, \xb,t} , 
\label{eq:reducedKet}
\end{equation}
\end{widetext}
where $\anglink = \{ i_1 < i_2 < \dots < i_k : i \in [1,n] \}$ denotes ordered subsets of $k$ indexes each in $\{ 1, 2, \dots, n \}$, and $\nu_k(\alb_\anglink, \xb_\anglink, t)$ are $k$-particle interaction functions up to a cutoff order $K$. 
We note that $\{ \dots \}_{k=1}^K$ is used to denote an index-ordered set in this context. 
This expansion of $n$-body interactions is a specific case of more general dimension-wise decompositions, such as analysis of least variance (ANOVA)~\cite{griebel_2005}. 
The probability of being in a state $\ket{n,\alb,\xb,t}$ is given by a dynamic and instantaneous Boltzmann distribution:
\begin{equation}
\pt (n,\alb,\xb,t) = \braket{n,\alb,\xb,t}{\{ \nu_k \}_{k=1}^K, t }
= \frac{
\exp [ - \sum_{k=1}^K \sum_\anglink \nu_k (\alb_\anglink, \xb_\anglink, t) ]
}{
\mathcal{Z} \left [ \{ \nu_k \}_{k=1}^K \right ]
}.
\label{eq:boltz}
\end{equation}

The true probability distribution $p(n,\alb,\xb,t)$ evolves according to the CME~(\ref{eq:cme}). 
To describe the time evolution of the reduced model, introduce a set of functionals $\{ {\cal F}_k \}_{k=1}^K$, forming a differential equation system for the interaction functions $\{ \nu_k \}_{k=1}^K$:
\begin{equation}
\frac{d}{dt} \nu_k (\alb_\anglink, \xb_\anglink, t) 
= {\cal F}_k \left [ \{ \nu(\alb, \xb,t) \} \right ] ,
\label{eq:ade}
\end{equation}
where
\begin{equation}
\{ \nu(\alb, \xb,t) \} = \left \{ \nu_\kp(\alb_\angljnkp, \xb_\angljnkp, t) \; \forall \; \angljnkp : 1 \leq \kp \leq K \right \}
\end{equation}
denotes all possible $\nu$ functions evaluated at the given arguments. 
Here, the right hand side ${\cal F}_k$ may be a global functional, in the sense that the arguments $\{ \nu(\alb, \xb,t) \}$ are not restricted to the arguments appearing on the left hand side of~(\ref{eq:ade}). 
We consider particular local parameterizations of this general form in Section~\ref{sec:2.F}.

In addition to the connection to MRFs, we note several advantages of the form of this reduced model~(\ref{eq:boltz},\ref{eq:ade}):
\begin{enumerate}
\item Since the states $\ket{ \{ \nu_k \}_{k=1}^K, t }$ define a grand canonical ensemble (GCE),~(\ref{eq:boltz}) exactly describes equilibrium systems, and is expected to reasonably approximate systems approaching equilibrium.
\item If the interactions between two groups of particles are independent, their joint probability distribution equals the product of their probabilities, 
and their interaction functions $\nu_k$ in~(\ref{eq:boltz}) sum. 
The Boltzmann distribution thus preserves the locality of interactions. 
\item A further important result pertains to linearity, stated in the following proposition.
\end{enumerate}

\begin{prop}
Given a reaction network and a fixed collection of $K$ interaction functions $\{\nu_k\}_{k=1}^K$, the linearity of the CME in reaction operators $\dot{p} = \sum_r \matr{W}^{(r)} p$ extends to the functionals ${\cal F}_k = \sum_r {\cal F}_k^{(r)}$.
\end{prop}

\begin{proof}
The dynamic Boltzmann distribution $\pt(n,\alb,\xb,t)$ is a maximum entropy (MaxEnt) distribution, where each interaction function $\nu_k(\alb_\anglink,\xb_\anglink,t)$ controls a corresponding moment $\mu_k(\alb_\anglink,\xb_\anglink,t) \}$, given by:
\begin{widetext}
\begin{equation}
\mu_k(\alb_\anglink, \xb_\anglink,t)
= \sum_{n^\prime=0}^\infty \sum_{\alb^\prime} \int d\xb^\prime \; p(n^\prime,\alb^\prime,\xb^\prime,t)
\sum_\angljnpk \delta(\xb_\anglink - \xb_\angljnpk^\prime) \delta(\alb_\anglink - \alb_\angljnpk^\prime) .
\label{eq:moment}
\end{equation}
\end{widetext}
Here, $\delta(\xb)$ denotes a multi-dimensional Dirac delta function. 
Note that there are $L = \sum_{k=1}^K \binom{n}{k}$ interaction terms and equally many moments they control. 
Switching to vector notation, let $\vnu$ of length $L$ denote the interaction functions, and $\vmu$ the corresponding moments.

Relating the interaction functions to the moments constitutes an inverse Ising problem. 
Let the solution to this problem be
\begin{equation}
\vnu_l = \phi_l ( \vmu )
\label{eq:invIsingCont}
\end{equation}
for some functions $\phi_l$ for $l=1,\dots,L$. 
This solution depends only on the interaction functions, and not on the reaction operators appearing in the CME. 
For a single reaction process, let the differential equations for the moments be $\dot{\vmu}^{(r)}$, resulting from $\dot{p}^{(r)} = \matr{W}^{(r)} p^{(r)}$, where $\dot{x}$ denotes a time derivative. 
Taking the time derivatives of both sides of~(\ref{eq:invIsingCont}) gives:
\begin{equation}
\dot{\vnu_l}^{(r)}
= 
\sum_{\lp=1}^L
\frac{ \partial \phi_l (\vmu) }{ \partial \vmu_\lp }
\dot{\vmu}_\lp^{(r)}
\end{equation}
For the full network of reactions then:
\begin{equation}
\dot{\vnu_l}
= 
\sum_{\lp=1}^L
\frac{ \partial \phi_l (\vmu) }{ \partial \vmu_\lp }
\dot{\vmu}_\lp
=
\sum_r
\sum_{\lp=1}^L
\frac{ \partial \phi_l (\vmu) }{ \partial \vmu_\lp }
\dot{\vmu}_\lp^{(r)}
=
\sum_r
\dot{\vnu}_l^{(r)}
\label{eq:nudot}
\end{equation}
gives the desired linearity property.

\end{proof}

Due to Proposition~1, the functionals ${\cal F}$ will be referred to as \textit{basis functionals}. In Section~\ref{sec:3.C}, the utility of this property is explored further in a machine learning context.


\subsection{\label{sec:2.C}Formulation of General Problem to Determine Functionals Governing Spatial Dynamics}

We next formulate a general problem to determine the functionals leading at all times to the MaxEnt Boltzmann distribution. 
Define the action as the KL-divergence between the true and reduced models (extending Ref.~\onlinecite{johnson_2015}):
\begin{equation}
\begin{split}
S &= \int_0^\infty dt \; \dkl(p || \pt) , \\
\dkl (p || \pt) &= \sum_{n=0}^\infty \sum_\alb \int d\xb \; p(n,\alb,\xb,t) \ln \frac{p(n,\alb,\xb,t)}{\pt (n,\alb,\xb,t)} .
\end{split}
\end{equation}
Next, we introduce notation to define a higher-order variational problem. 
Since the interaction functions are defined by specifying the set of functionals ${\cal F}_k [ \{ \nu (\alb,\xb,t)\} ]$ for all $k =1,\dots,K$, we use the notation $\nu_k \mbox{{\textlbrackdbl}} \{ {\cal F} \} \mbox{{\textrbrackdbl}}$ to denote that $\nu_k$ is a higher-order generalization of a functional. 
The action is a functional of the set of all interaction functions, which we denote by $S [ \{ \nu \mbox{{\textlbrackdbl}} \{ {\cal F} \} \mbox{{\textrbrackdbl}} \} ]$, where $\{ x \} = \{ x_k \}_{k=1}^K$. 

The higher-order variational problem for the basis functionals is given by the chain rule:
\begin{widetext}
\begin{equation}
\frac{
\deltaHI 
S [ \{ \nu \mbox{{\textlbrackdbl}} \{ {\cal F} \} \mbox{{\textrbrackdbl}} \} ]
}
{\deltaHI {\cal F}_k \left [ \{ \nu(\alb, \xb,t) \} \right ] }
=
\sum_{\kp=1}^K \sum_{\alb^\prime} \int d\xb^\prime \int d\tp \; \frac{\delta 
S [ \{ \nu \} ]
}{
\delta \nu_{\kp}(\alb^\prime,\xb^\prime,\tp)
} \frac{\deltaHI \nu_{\kp}(\alb^\prime,\xb^\prime,\tp)}{\deltaHI {\cal F}_k \left [ \{ \nu(\alb, \xb,t) \} \right ] }
=
0
,
\label{eq:varProblemGeneralChain}
\end{equation}
\end{widetext}
where we use the notation $\deltaHI$ to denote that this is not an ordinary variational problem, in the sense that a variation with respect to a functional is implied. 
Equation~(\ref{eq:varProblemGeneralChain}) should therefore be regarded as a purely notation solution, generalizing the well-known chain rule for functionals where a variational derivative is taken of a functional of a functional: $\frac{\delta F[G[\phi]]}{\delta \phi(y)} = \int dx \; \frac{\delta F[G]}{\delta G(x)} \frac{\delta G[\phi]}{\delta \phi(y)}$. 
The first term is a variational derivative analogous to that appearing in the derivation of the BM learning algorithm~\cite{ackley_1985}, giving:
\begin{widetext}
\begin{equation}
\frac{
\deltaHI 
S [ \{ \nu \mbox{{\textlbrackdbl}} \{ {\cal F} \} \mbox{{\textrbrackdbl}} \} ]
}
{\deltaHI {\cal F}_k \left [ \{ \nu(\alb, \xb,t) \} \right ] }
=
\sum_{\kp=1}^K \sum_{\alb^\prime} \int d\xb^\prime \int_0^\infty d\tp \; 
\left ( 
\mu_\kp(\alb^\prime, \xb^\prime,\tp)
- \tilde{\mu}_\kp(\alb^\prime, \xb^\prime,\tp)
\right ) 
\frac{\deltaHI \nu_{\kp}(\alb^\prime,\xb^\prime,\tp)}{\deltaHI {\cal F}_k \left [ \{ \nu(\alb, \xb,t) \} \right ] } 
= 0 ,
\label{eq:varProblemGeneral}
\end{equation}
\end{widetext}
where the moments $\mu$ are defined in~(\ref{eq:moment}), with $\tilde{\mu}$ having $p$ replaced by $\pt$. Next, we consider well-mixed systems where the de-escalation from functionals to ordinary functions $F_k$ makes this problem~(\ref{eq:varProblemGeneral}) well-defined. In Section~\ref{sec:2.F}, we parameterize the functional form of ${\cal F}$ to consider spatially distributed systems.


\subsection{\label{sec:2.D}Learning Algorithm for Reduced Dynamics of Well-Mixed Systems in One Species}

In the case of well-mixed systems in one species, the state of the system is entirely characterized by the number of individuals $\ket{n,t}$. Dropping the species and position labels in the dynamic Boltzmann distribution gives:
\begin{align}
\pt (n,t) = \frac{1}{\mathcal{Z} \left ( \{ \nu \} \right )} \exp [ - \sum_{k=1}^K \binom{n}{k} \nu_k (t) ] ,
\label{eq:boltzWM}
\end{align}
where we use the notation $\{ \nu \} = \{ \nu_{\kp} \}_{\kp=1}^K$. The time evolution is now described by basis functions forming the autonomous differential equation system:
\begin{equation}
\begin{split}
\frac{d}{dt} \nu_k (t) &= F_k \left ( \{ \nu \} \right ), \\
\text{with I.C.:}\quad \nu_k(t=0) &= \eta_k ,
\end{split}
\label{eq:adeWM}
\end{equation}
where $F_k$ are now functions rather than functionals ${\cal F}_k$. 
The variational problem~(\ref{eq:varProblemGeneral}) for the basis functions becomes:
\begin{widetext}
\begin{equation}
\frac{\delta S [ \{ \nu [ \{ F \} ] \} ] }{\delta F_k( \{ \nu \} )}
=
\sum_{\kp=1}^K \int_0^\infty d\tp \; \left ( \left \langle \binom{n}{\kp} \right \rangle_{p (\tp)} - \left \langle \binom{n}{\kp} \right \rangle_{\pt (\tp)} \right ) \frac{\delta \nu_{\kp}(\tp)}{\delta F_k( \{ \nu \} )} 
= 0 ,
\label{eq:varProblem}
\end{equation}
\end{widetext}
where $\left \langle X \right \rangle_p(\tp) = \sum_{n=0}^\infty X p(n,\tp)$ and similarly for $\pt$.

The variational term on the RHS of~(\ref{eq:varProblem}) may be determined by a number of methods, including by an ODE formulation derived in Appendix~\ref{app:varTerm:WM}, a PDE formulation derived from applying the chain rule at the initial condition, and using a Lie series approach (Supplemental Material).
The first of these and arguably the most practical is:
\begin{equation}
\begin{split}
\frac{d}{d\tp} \left ( \frac{\delta \nu_\kp(\tp)}{\delta \fknu} \right )
=&
\sum_{l=1}^K
\frac{\partial F_\kp(\{ \nu(\tp) \})}{\partial \nu_l(\tp)}
\frac{\delta \nu_l(\tp)}{\delta F_k (\{ \nu \})} 
+ \delta_{\kp,k} \delta ( \{ \nu \} - \{ \nu(\tp) \} ) , \\
\text{with I.C.:}\quad \frac{\delta \nu_\kp(\tp=0)}{\delta \fknu } =& 0 .
\end{split}
\label{eq:pdeWM}
\end{equation}
An algorithmic solution to~(\ref{eq:varProblem}) is therefore possible in the form of a PDE-constrained optimization problem: 
Solve~(\ref{eq:varProblem},\ref{eq:pdeWM}) subject to the PDE-constraint~(\ref{eq:adeWM}). 
An example algorithm using simple gradient descent is given by Algorithm~\ref{alg:1}.

\begin{figure}[t]
\begin{algorithm}[H]
\caption{Gradient Descent for Learning Basis Functions Governing Well-Mixed Dynamics} \label{alg:1}
\begin{algorithmic}[1]
\algsetblock[Init]{Initialize}{Stop}{4}{1.4em}
\Initialize
\State Grid of values $\{ \nu \}$ to solve over.
\State $\fknu$ for $k=1,\dots,K$.
\State Max. integration time $T$.
\State A formula for the learning rate $\lambda$.
\While{not converged}
	\State Initialize $\Delta F_k(\{\nu\}) = 0$ for all $k,\{\nu\}$. \;
    \State Generate a sample of random initial conditions $\{ \eta \}$. \;
    \For{$\eta_i \in \{ \eta \}$}
		\LineComment{\textit{Generate trajectory in reduced space $\{ \nu \}$:}}
		\State Solve the PDE constraint~(\ref{eq:adeWM}) with IC $\eta_i$ for $0 \leq t \leq T$. \;
		\State Solve~(\ref{eq:pdeWM}) for variational term $\delta \nu_\kp(t)/\delta \fknu$. \;
		\LineComment{\textit{Sampling step:}}
		\State Evaluate moments $\left \langle \binom{n}{\kp} \right \rangle_{\pt (\tp)}$ of the Boltzmann distribution by sampling or analytically. \;
		\State Evaluate true moments $\left \langle \binom{n}{\kp} \right \rangle_{p (\tp)}$ by stochastic simulation or analytic solution. \;
        \LineComment{\textit{Evaluate the objective function:}}
        \State Update $\Delta F_k(\{\nu\})$ as the cumulative moving average of~(\ref{eq:varProblem}) over initial conditions $\{ \eta \}$. \;
	\EndFor
    \LineComment{\textit{Update to decrease objective function:}}
	\State Update $\fknu$ to decrease the objective function: $\fknu \rightarrow \fknu - \lambda \Delta F_k(\{\nu\} )$. \;
\EndWhile
\end{algorithmic}
\end{algorithm}
\end{figure}

We note the implicit connection between this approach and using Boltzmann machines, such as in GCCD, by the algorithm's objective function. 
Here, the whole trajectory of moments from stochastic simulations is used to directly estimate time evolution operators, rather than estimating the interaction parameters at each time step.
We make this connection explicit in Algorithm~2 in Section~\ref{sec:3.C} below.

Further improvements to Algorithm~1 are possible, such as to replace ordinary gradient descent by an accelerated version, e.g. Nesterov accelerated gradient descent~\cite{nesterov_1983}. 
Furthermore, the wealth of methods available to solve PDE-constrained optimization problems, e.g. adjoint methods~\cite{giles_2000}, offer rich possibilities for further development.


\subsubsection{\label{sec:2.D.1}Example: Mean of the Galton-Watson Branching Process}
 
As a simple illustrative example, consider a reduced model that captures the time-evolving mean of the Galton-Watson branching process, consisting of the birth process $A \rightarrow A+A$ with rate $k_b$ and death $A \rightarrow \oset$ with rate $k_d$.

In this case, there are only self-interactions ($K=1$) described by $\nu(t)$ with basis function $F(\nu(t))$. 
The dynamic Boltzmann distribution is:
\begin{equation}
\pt(n,t) = \frac{1}{\mathcal{Z}} \exp[-n \nu(t)] . 
\end{equation}
Using the fact that
\begin{equation}
\langle n \rangle_{\pt} = \frac{1}{e^{\nu} - 1}
\end{equation}
and from the CME
\begin{equation}
\frac{d \langle n \rangle_p}{dt} = (k_b-k_d) \langle n \rangle_p
\end{equation}
gives the analytic solution for the basis functions
\begin{equation}
F(\nu) = (k_b-k_d) ( e^{-\nu} - 1 ) . 
\end{equation}

This solution is reproduced using Algorithm~1, as shown in Figure~\ref{fig:ex1:sol} for $k_d = 3 k_b / 2$. 
Here, the solution is constructed on a grid of $\nu \in [0^+,3.0]$ with spacing $\Delta \nu = 0.1$, with maximum integration time $T=1$ (arbitrary units). 
The learning rate is decreased exponentially over iterations to improve convergence. 
The convergence of the algorithm is shown in Figure~\ref{fig:ex1:sanddsdf}.

\begin{figure}[!ht]
	\centering
	\includegraphics[width=0.5\columnwidth]{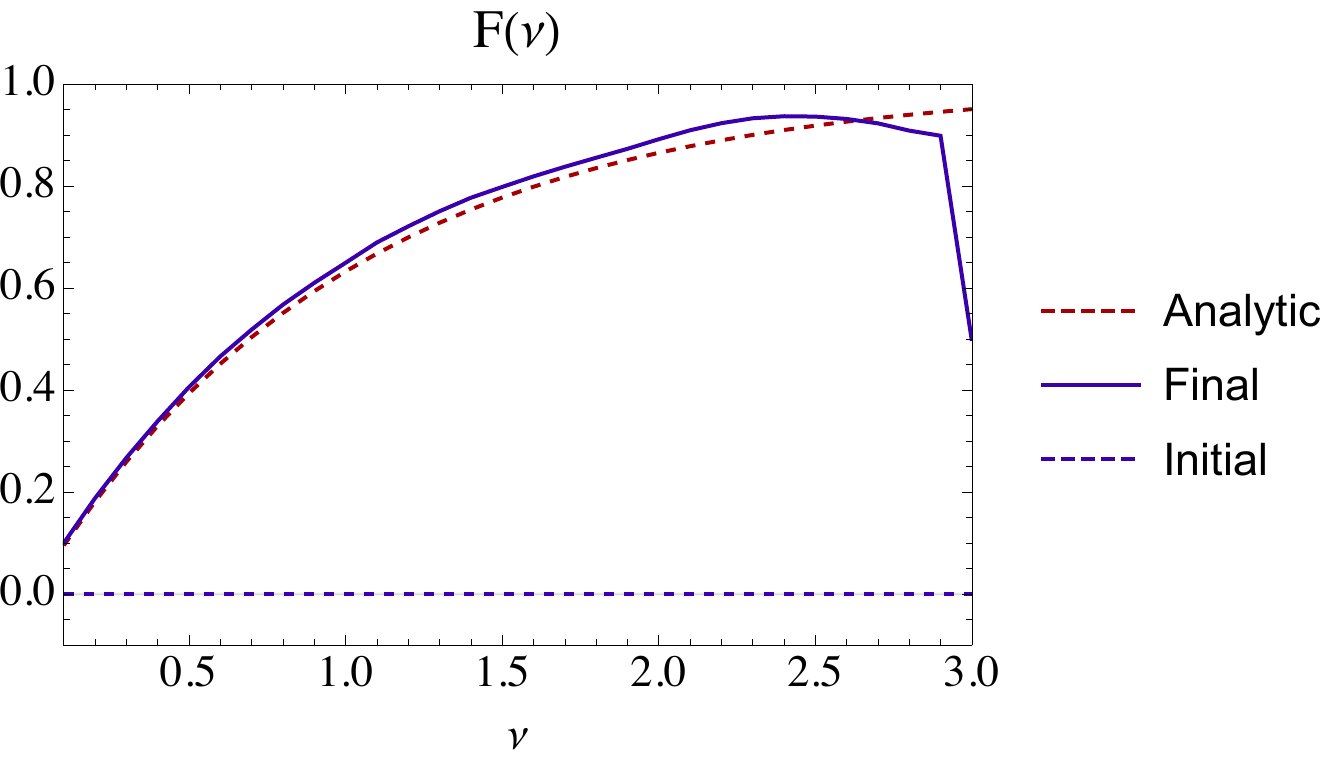} 
	\caption{Learned basis function for the simple annihilation process $A\rightarrow\oset$ after 40 iterations, from a uniform initial condition.}
	\label{fig:ex1:sol}
\end{figure}


\begin{figure}[!ht]
	\centering
	\includegraphics[width=0.7\columnwidth]{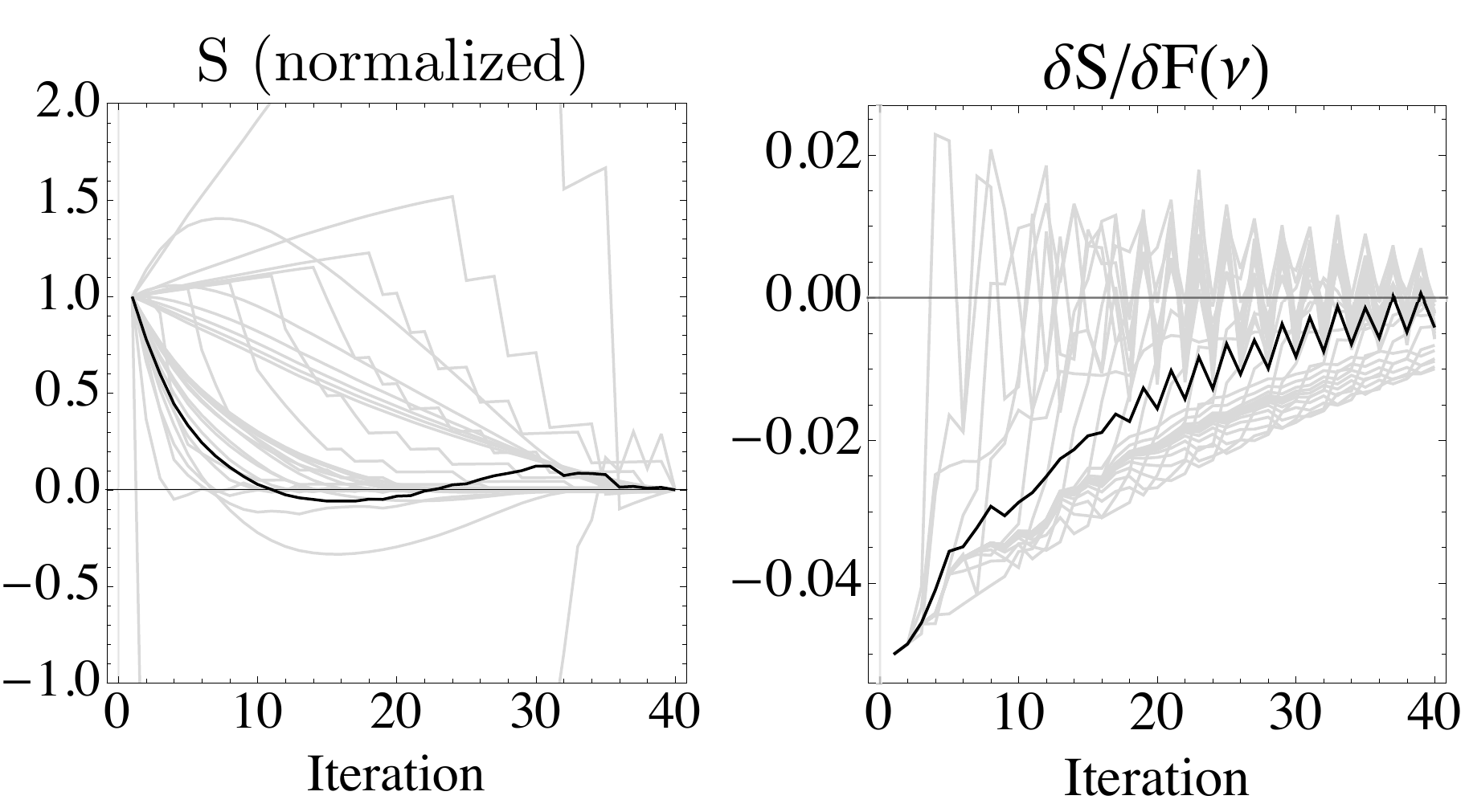}
	\caption{\textit{Left:} The convergence of the action $S$ as it is minimized over iterations following Algorithm~1. Trajectories (grey) for individual $\eta$, normalized to start at one and end at zero, and their mean (black). \textit{Right:} The minimization of the variation in the action $\delta S/\delta F(\nu)$, as a function of the position to vary $\nu$. Trajectories (grey) at each $\nu$, and their mean (black).}
	\label{fig:ex1:sanddsdf}
\end{figure}

\begin{figure}[!ht]
	\centering
	\includegraphics[width=0.5\columnwidth]{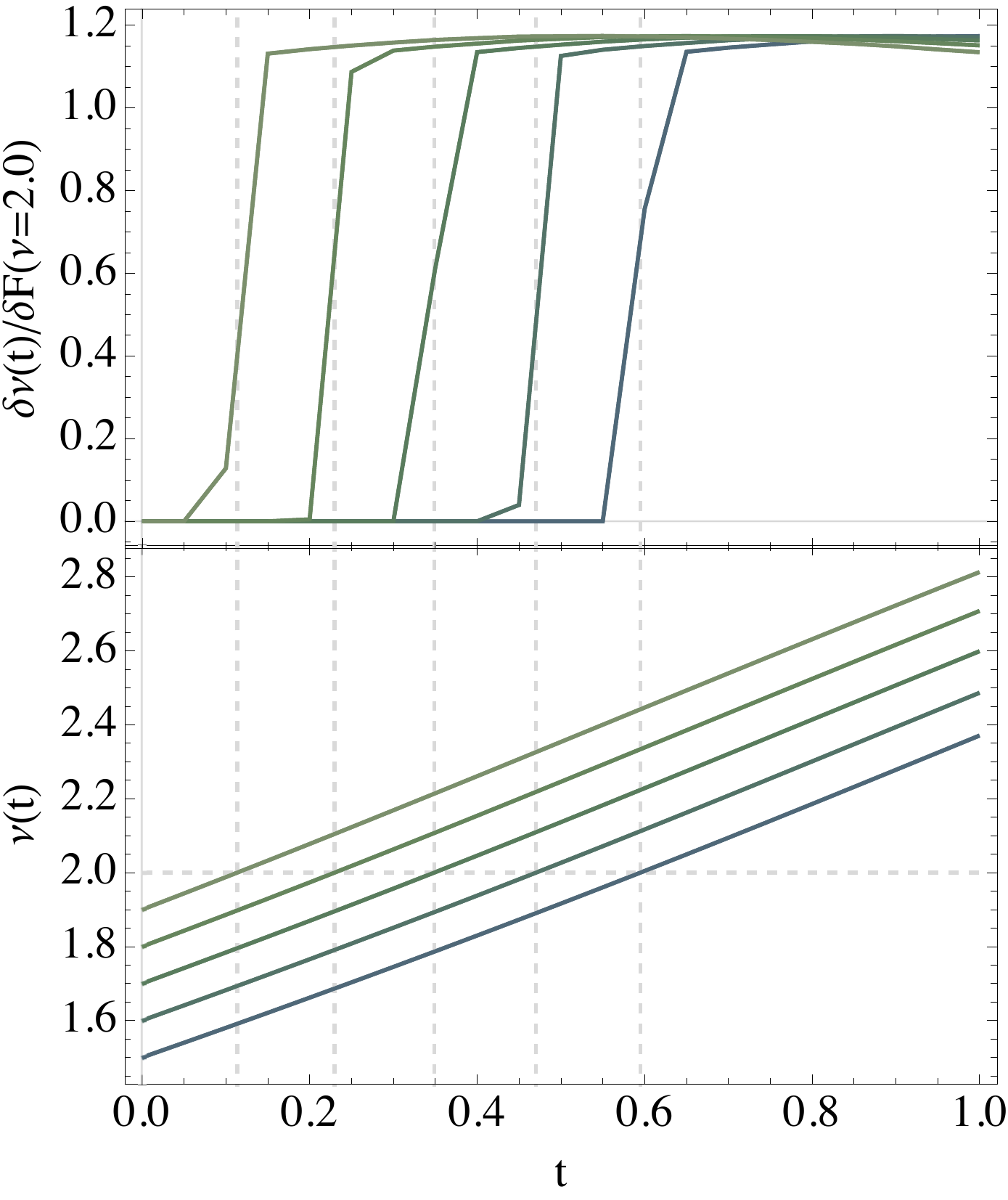}
	\caption{\textit{Top:} The variational term $\delta \nu(t)/\delta F(\nu=2.0)$ for several initial conditions $\eta=1.5,\dots, 1.9$ as a function of time, obtained by solving~(\ref{eq:pdeWM}) numerically. \textit{Bottom:} The solution trajectories $\nu(t)$ starting from these $\eta$. Only when the solution trajectory is close to $\nu(t)=2.0$ and thereafter does varying the basis function $F$ at this point have a non-zero effect.}
	\label{fig:ex1:r}
\end{figure}


\subsubsection{\label{sec:2.D.2}Example: Two Basis Functions Controlling Mean and Variance}

Consider again the process of the previous section, but with $K=2$ basis functions $\nu_1(t)$ and $\nu_2(t)$ controlling the mean and variance in the number of particles. 
The dynamic Boltzmann distribution is:
\begin{equation}
\pt(n,t) = \frac{1}{\mathcal{Z}} \exp[-n \nu_1(t) - \binom{n}{2} \nu_2(t)] .
\end{equation}

This may be interpreted as a Gau{\ss}ian distribution in the number of particles, provided we treat $n$ as continuous and extend its range to $\pm \infty$, or consider systems with means far from $n=0$. 
In this case, the mean $\mu$ and variance $\sigma^2$ can be related to the interaction functions as $\mu = 1/2 - \nu_1/\nu_2$ and $\sigma^2 = 1 / \nu_2$. 
The differential equations derived from the CME for the moments of this system are:
\begin{equation}
\begin{split}
\frac{d \mu}{dt} &= (k_b-k_d) \mu , \\
\frac{d \sigma^2}{dt} &= 2 (k_b-k_d) \sigma^2 + (k_b+k_d) \mu ,
\end{split}
\end{equation}
which can be converted to analytic solutions for the basis functions:
\begin{equation}
\begin{split}
F_1(\nu_1,\nu_2) = &\nu_1 \left( k_d - k_b + (k_b + k_d) \nu_1 \right ) - \frac{\nu_2}{2} \left ( k_b - k_d + (k_b + k_d) \nu_1 \right ) ,
\\
F_2(\nu_1,\nu_2) = &- \frac{\nu_2}{2} \Big ( k_d \left ( - 4 - 2 \nu_1 + \nu_2 \right ) + k_b \left ( 4 - 2 \nu_1 + \nu_2 \right ) \Big ) .
\end{split}
\label{eq:gw}
\end{equation}
These are shown in Figure~\ref{fig:gw_2d_bfs}. 

Figure~\ref{fig:gw_2d_var_term} shows the variational terms $\delta \nu_1(t)/\delta F_1(\nu_1,\nu_2)$ and $\delta \nu_2(t)/\delta F_1(\nu_1,\nu_2)$, resulting from Algorithm~1 and determined by~(\ref{eq:pdeWM}). 
Interestingly, the self-varying term $\delta \nu_1(t)/\delta F_1(\nu_1,\nu_2)$ more closely resembles the multivariate delta-function appearing in~(\ref{eq:pdeWM}), while the cross term $\delta \nu_2(t)/\delta F_1(\nu_1,\nu_2)$ shows a greater temporal memory of the solution trajectory. 


\begin{figure}[!ht]
	\centering
	\includegraphics[width=0.7\columnwidth]{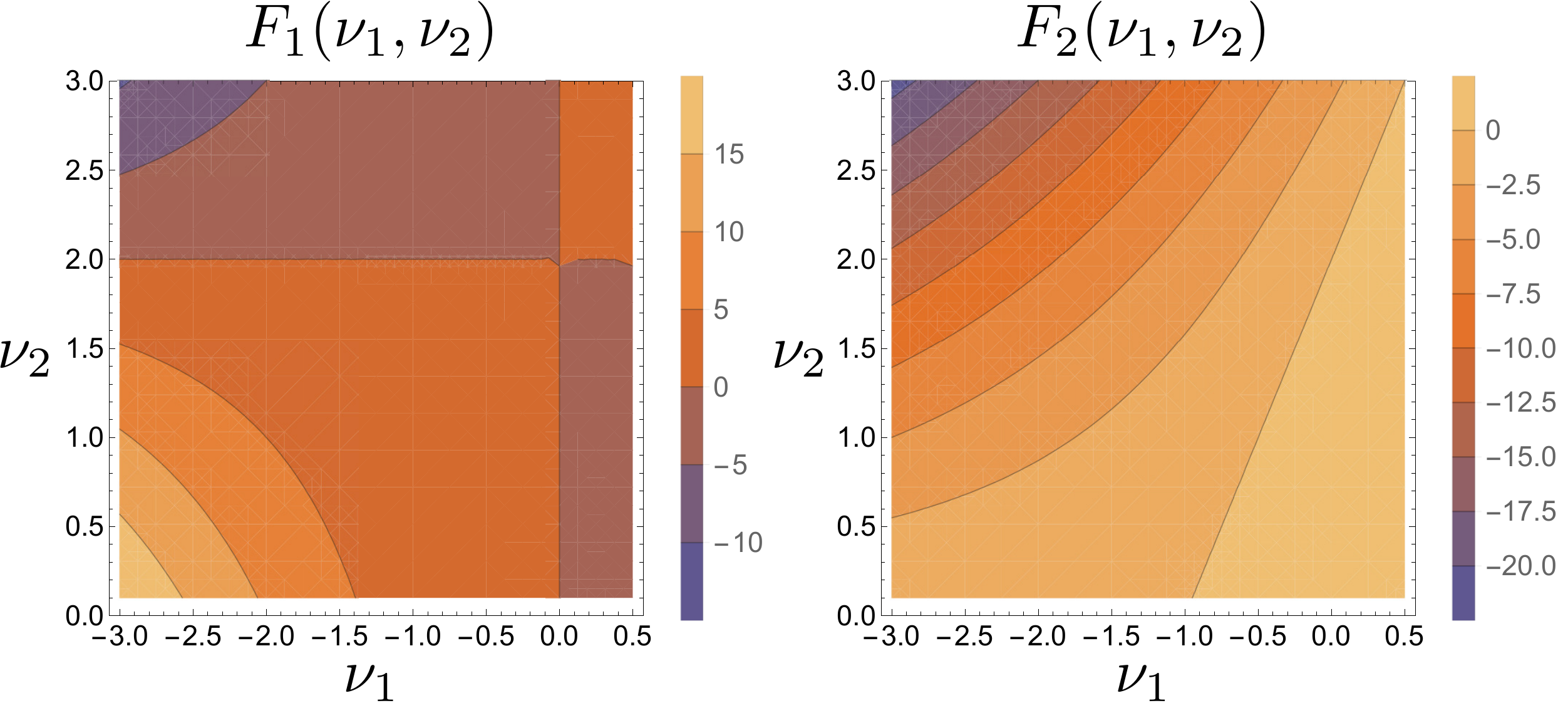}
	\caption{The true basis functions~(\ref{eq:gw}) for the Galton-Watson system. \textit{Left:} $F_1(\nu_1,\nu_2)$. \textit{Right:} $F_2(\nu_1,\nu_2)$. The reaction rates used are $k_d = 3 k_b / 2 = 1.5$.}
	\label{fig:gw_2d_bfs}
\end{figure}


\begin{figure}[!ht]
	\centering
	\includegraphics[width=0.7\columnwidth]{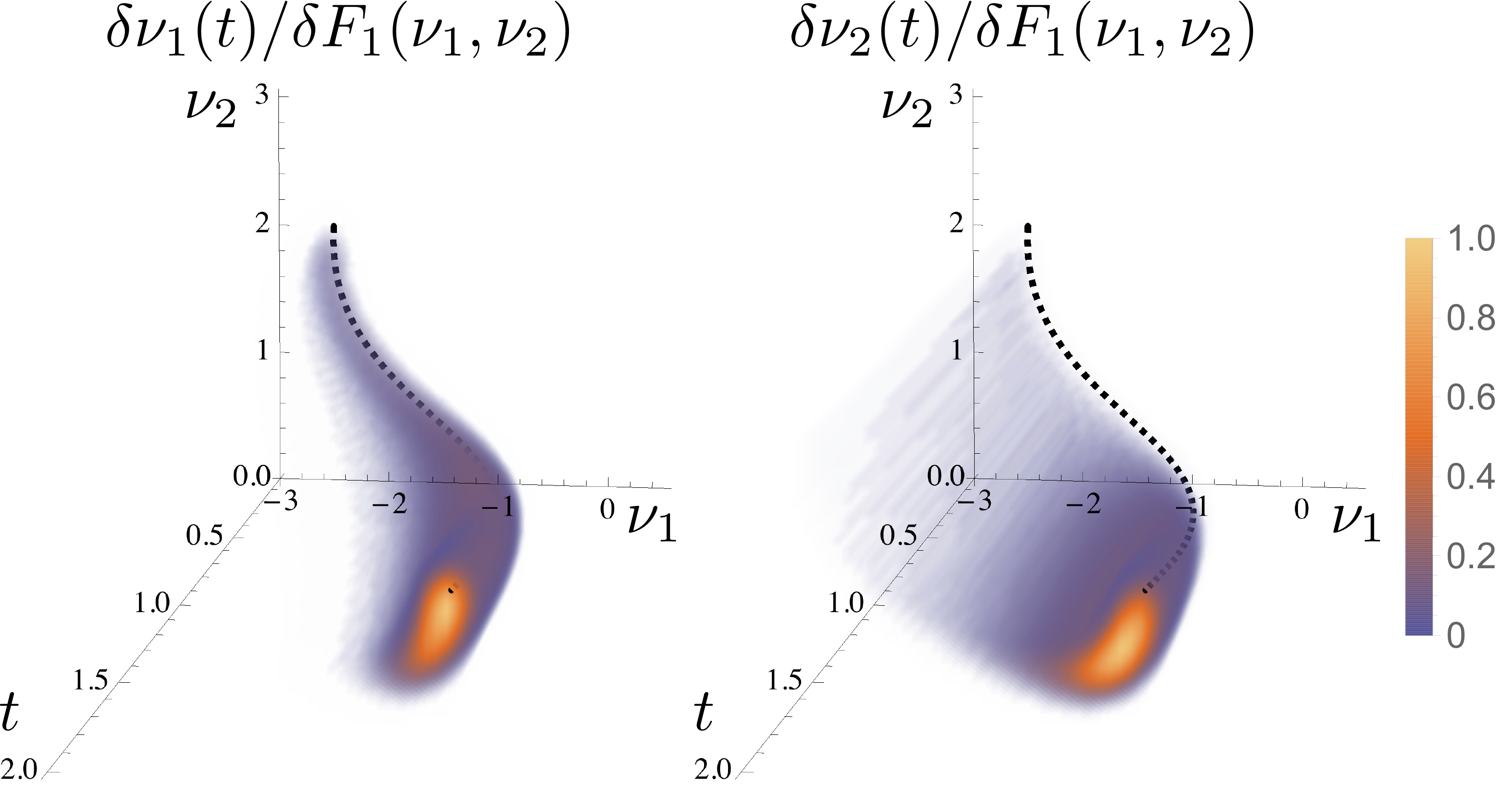}
	\caption{Two of the four variational terms for the Galton-Watson system using two basis functions. \textit{Left:} $\delta \nu_1(t)/\delta F_1(\nu_1,\nu_2)$. \textit{Right:} $\delta \nu_2(t)/\delta F_1(\nu_1,\nu_2)$. The black dashed line shows the solution trajectory. The initial conditions are $(\eta_1,\eta_2)=(-2.5,2.0)$, with reaction rates $k_d = 3 k_b / 2 = 1.5$. The effect of the mixing term $\delta \nu_2(t)/\delta F_1(\nu_1,\nu_2)$ is lower in magnitude but persists longer over time. Note that the absolute magnitude of the spread is related to the approximation chosen for the delta function in~(\ref{eq:pdeWM}), a normalized multivariate Gau{\ss}ian with variance of 0.1 in both directions $\nu_1,\nu_2$.}
    \label{fig:gw_2d_var_term}
\end{figure}


\subsection{\label{sec:2.E}Analytic MaxEnt Solutions}

We next consider special cases where analytic solutions for the basis functionals are possible, to motivate a parameterization leading to a solvable
version of the variational problem~(\ref{eq:varProblemGeneral}).


\subsubsection{\label{sec:2.E.1}Gau{\ss}ian Distributions}

The well-mixed case~(\ref{eq:boltzWM}) is the MaxEnt distribution consistent with $\left \langle \binom{n}{k} \right \rangle_p$ for $k=1,\dots,K$. 
If $K=2$, then~(\ref{eq:boltzWM}) may be interpreted as a Gau{\ss}ian distribution in continuous $n$, as discussed in the previous section. 
Generalizing these results, the basis functions are generally given by:
\begin{equation}
\begin{split}
F_1(\nu_1,\nu_2) &= - \nu_2 \frac{d \mu}{dt} - \nu_1 \nu_2 \frac{d \sigma^2}{dt} , \\
F_2(\nu_1,\nu_2) &= - \nu_2^2 \frac{d \sigma^2}{dt} ,
\end{split}
\end{equation}
where $d \mu/dt$, $d\sigma^2/dt$ are evaluated from the CME and expressed in terms of $\nu_1,\nu_2$. 
Here, a moment closure approximation must be applied if the reactions are greater than unimolecular in number of reagents. 
For example, the higher order moments appearing in the CME may be approximated by those of the reduced model $\pt$ and expressed in terms of lower order $\mu,\sigma^2$ following the well known property of Gau{\ss}ian distributions. 
This closure technique is described further in Section~\ref{sec:3.A}.


\subsubsection{\label{sec:2.E.2}Diffusion from Point Source}

In the spatial case, consider a diffusion process of a fixed number of particles $n$ with diffusion constant $D$ spreading out from a point source at $x_0$. 
The analytic solution to the CME is:
\begin{equation}
p(\xb,t) = \left ( 4\pi Dt \right )^{-n/2} \exp [ - \sum_{i=1}^{n} \frac{(x_i-x_0)^2}{4Dt} ] ,
\label{eq:diff}
\end{equation}
reflecting that only self interactions ($K=1$) are necessary to describe the process. 
The reduced model~(\ref{eq:boltz}) becomes:
\begin{equation}
\pt(\xb,t) = \frac{1}{\mathcal{Z}} \exp [ - \sum_{i=1}^{n} \ovnu(x_i,t) ] .
\label{eq:boltzDiff}
\end{equation}
It is straightforward to verify that $p(\xb,t) = \pt(\xb,t)$ if
\begin{equation}
\ovnu(y,t) = \ln \left ( 1 + \frac{1}{n} \right ) + \frac{1}{2} \ln \left (4 \pi D t \right ) + \frac{(y-x_0)^2}{4Dt} .
\label{eq:diffInvIsing}
\end{equation}
Consequentially, from $\partial_t \ovnu(y,t)$, the basis functional is:
\begin{equation}
F [ \ovnu(y,t) ] = D \partial_y^2 \ovnu(y,t) - D \left ( \partial_y \ovnu(y,t) \right )^2 .
\label{eq:bfdiff}
\end{equation}


\subsubsection{\label{sec:2.E.3}Unimolecular Reaction-Diffusion}

For reaction networks that involve only diffusion and unimolecular reactions, two key properties hold for the CME solution:
\begin{enumerate}
\item Separable spatial and particle number distributions $p(n,\xb) = p(n) p(\xb)$ where each distribution is normalized $\sum_n p(n) = 1$ and $\int d\xb \; p(\xb) = 1$.
\item Independence of spatial distribution $p(\xb) = p(x_1) p(x_2) \dots p(x_n)$ where each $\int dx \; p(x) = 1$ is normalized. 
This assumes that initial $p(x_i)$ are independent - otherwise, a fixed mixture of independent components must be considered. 
\end{enumerate}

Analogous to the purely diffusive process above, this allows analytic solutions to the inverse Ising problem by imposing these conditions upon the dynamic Boltzmann distribution $\pt$. 
Here, we exploit the fact that multiplication of Boltzmann distributions results in addition of the energy functions.

Introduce a single interaction function $\ovnu(x,t)$ to capture the diffusion process and the usual $\nu_1(t), \dots, \nu_K(t)$ to describe the reactions (for brevity, omit further time arguments in this section). 
Furthermore, impose the normalization $\int dx \exp [-\ovnu(x)] = 1$. 
The dynamic Boltzmann distribution becomes:
\begin{equation}
\begin{split}
\pt(n,\xb) &= \pt(n) \pt(\xb) = \pt(n) \pt(x_1) \pt(x_2) \dots \pt(x_n) , \\
\pt(n) &= \frac{1}{\mathcal{Z}} \exp [ - \sum_{k=1}^K \binom{n}{k} \nu_k ] , \\
\pt(x) &= \exp [ -\ovnu(x) ] ,
\end{split}
\end{equation}
where the partition function is
\begin{equation}
\mathcal{Z} = \sum_n \int d\xb \; \pt(n,\xb) = \sum_n \pt(n) .
\end{equation}

The distribution $\pt(n,\xb)$ is the MaxEnt distribution consistent with the moments $\langle \binom{n}{k} \rangle$ for all $k=1,\dots,K$, as well as the spatial moment:
\begin{equation}
\left \langle \sum_{i=1}^n \delta(y-x_i) \right \rangle = \sum_n \int d\xb \; \sum_{i=1}^n \delta(y-x_i) \pt(n,\xb) = \exp [-\ovnu(y) ] \left \langle n \right \rangle ,
\end{equation}
such that the solution to the inverse Ising problem is:
\begin{equation}
\ovnu(y) = \ln \left ( 
\frac{ \left \langle n \right \rangle }{ \left \langle \sum_{i=1}^n \delta(y-x_i) \right \rangle }
\right ) .
\end{equation}

The solution for the inverse Ising problem for $\langle \binom{n}{k} \rangle$ is independent of this spatial moment, and analytically possible e.g. for $K=1$ or $2$, as demonstrated in Sections~\ref{sec:2.D.1} and \ref{sec:2.E.1} above.

Taking the time derivatives of these solutions $\dot{\ovnu}$ and $\dot{\nu_k}$ and using the CME to derive differential equations for the moments gives the basis functionals. 
For unimolecular reactions, the diffusion process does not affect the reactions, such that the functional controlling $\ovnu$ is always that of diffusion~(\ref{eq:bfdiff}). 
For example, for a branching random walk consisting of diffusion from a point source and the Galton-Watson process with $K=2$, the basis are the functional~(\ref{eq:bfdiff}) and the functions~(\ref{eq:gw}).


\subsection{\label{sec:2.F}Parameterizations for Spatially Heterogeneous Systems}

For reaction-diffusion systems that involve reactions greater than unimolecular in number of reagents, it generally becomes difficult to analytically solve the inverse Ising problem and consequentially identify basis functionals.
However, an algorithmic solution remains possible, where we guess a local parameterization of the functional~(\ref{eq:ade}) based on the analytic solutions presented above.

Let $\beb,\yb$ be of length $k$, and use the notation
\begin{equation}
\{ \nu(\beb, \yb,t) \} = \left \{ \nu_{\kp} (\beb_\anglikkp, \yb_\anglikkp, t) \; \forall \; \anglikkp : 1 \leq \kp \leq k \right \} .
\end{equation}
Then choose the {\it spatially local} parameterization of ${\cal F}_k$ in (\ref{eq:ade}): 

\begin{widetext}
\begin{equation}
\begin{split}
\frac{d}{dt} \nu_k (\beb, \yb, t) 
=&
F_k^{(0)} (\{ \nu(\beb,\yb,t) \}) + \sum_{\lambda=1}^k \Bigg (
F_k^{(1,\lambda)} (\{ \nu(\beb,\yb,t) \}) \sum_\anglikl \sum_{m=1}^\lambda 
\left ( \partial_m \nu_\lambda(\beb_\anglikl,\yb_\anglikl,t) \right )^2 
\\ & \hspace{40mm}
+ 
F_k^{(2,\lambda)} (\{ \nu(\beb,\yb,t) \}) \sum_\anglikl \sum_{m=1}^\lambda 
\partial_m^2 \nu_\lambda(\beb_\anglikl,\yb_\anglikl,t) 
\Bigg ) ,
\\
\text{with I.C.:}\quad \nu_k(\beb, \yb,t=0) =& \eta_k(\beb, \yb) ,
\end{split}
\label{eq:adeLocal}
\end{equation}
\end{widetext}
where $\partial_m$ denotes the derivative with respect to the $m$-th component of $\yb_\anglikl$, and $F_k^{(\gamma)} (\{ \nu(\beb, \yb,t) \})$ for $(\gamma)=(0),(1,\lambda),(2,\lambda)$ are local functions, i.e. functions of the arguments on the left hand side of~(\ref{eq:adeLocal}).

The variational problem~(\ref{eq:varProblemGeneral}) now becomes
\begin{widetext}
\begin{equation}
\frac{\delta S[\{ \nu [\{ F \} ] \}] }{\delta F_k^{(\gamma)} (\{ \nu(\beb,\yb) \}) } 
=
\sum_{\kp=1}^K \sum_{\beb^\prime} \int d\yb^\prime \; \int d \tp \; 
\left ( 
\mu_\kp(\beb^\prime,\yb^\prime,\tp)
-
\tilde{\mu}_\kp(\beb^\prime,\yb^\prime,\tp)
\right ) 
\frac{
\delta \nu_\kp (\beb^\prime, \yb^\prime, \tp)
}{
\delta F_k^{(\gamma)} (\{ \nu(\beb,\yb) \}) }
=
0
\label{eq:optSpatial}
\end{equation}
\end{widetext}
for $(\gamma)=(0),(1,\lambda),(2,\lambda)$, where  $\beb^\prime,\yb^\prime$ are of length $\kp$.

Analogously to the well-mixed case, it is possible to derive a PDE system governing the variational term $\delta \nu_\kp (\beb^\prime, \yb^\prime, \tp) / \delta F_k^{(\gamma)} (\{ \nu(\beb,\yb) \})$. 
In Appendix~\ref{app:varTerm:diff}, an illustrative example is derived for a diffusion process.

Equations~(\ref{eq:adeLocal},\ref{eq:optSpatial}) together form a PDE-constrained optimization problem, which may be solved analogously to Algorithm~1, with additional spatial axes.


%

\subsubsection{\label{sec:2.F.1}Example: Branching Random Walk}

Consider a branching random walk consisting of the Galton-Watson process and diffusion from a point source with rate $D$ in one spatial dimension and one species. 
From the true solutions for the basis functionals~(\ref{eq:bfdiff},\ref{eq:gw}), use one spatial interaction function $\ovnu(y,t)$ and two purely temporal $\nu_1(t),\nu_2(t)$, and further restrict the parameterization~(\ref{eq:adeLocal}) of the basis functionals to be
\begin{align}
\begin{split}
\frac{d \ovnu(y,t)}{dt} =& \ovf[\ovnu(y,t)] =
\ovf^{(1)} (\ovnu(y,t)) \left ( \partial_y \ovnu(y,t) \right )^2 
+ \ovf^{(2)} (\ovnu(y,t)) \partial_y^2 \ovnu(y,t) ,
\end{split}
\label{eq:optc1}
\\
\frac{d \nu_k(t)}{dt} =& F_k[\nu_1(t),\nu_2(t)] = F_k^{(0)} (\nu_1(t),\nu_2(t))
\label{eq:optc2}
\end{align}
for $k=1,2$. 
The variational problem is
\begin{widetext}
\begin{align}
\frac{\delta S}{\delta \ovf^{(\gamma)} (\ovnu)}
=&
\int d\yp \int d\tp \; 
\left ( 
\mu_1(\yp,\tp) 
- \tilde{\mu}_1(\yp,\tp)
\right )
\frac{
\delta \ovnu (\yp, \tp)
}{
\delta \ovf^{(\gamma)} (\ovnu)
}
=
0 ,
\label{eq:opt1} \\
\frac{\delta S}{\delta F_k^{(0)} (\nu_1,\nu_2)}
=&
\sum_{\kp=1}^2
\int d\tp \; 
\left ( 
\left \langle \binom{n}{\kp} \right \rangle_{p(\tp)}
-
\left \langle \binom{n}{\kp} \right \rangle_{\pt(\tp)}
\right )
\frac{
\delta \nu_\kp (\tp)
}{
\delta F_k^{(0)} (\nu_1,\nu_2)
}
=
0
\label{eq:opt2}
\end{align}
\end{widetext}
for $\gamma=1,2$.

Differential equations governing the variational terms $\delta \ovnu(\yp,\tp) / \delta \ovf^{(\gamma)} (\ovnu)$ for $\gamma = 1,2$ are derived in Appendix~\ref{app:varTerm:diff}, given by~(\ref{eq:app:varTermsDiff}). 
Differential equations governing $\delta \nu_\kp(\tp) / \delta F_k^{(0)} (\nu_1,\nu_2)$ are given by~(\ref{eq:pdeWM}).

The optimization problem~(\ref{eq:opt1},\ref{eq:opt2}) subject to the PDE-constraints~(\ref{eq:optc1},\ref{eq:optc2}) may be solved algorithmically using Algorithm~1 in each $F_k,\ovf^{(1)},\ovf^{(2)}$, analogously to the well-mixed case. 
We note that the true solutions are given by~(\ref{eq:bfdiff},\ref{eq:gw}), in particular: $\ovf^{(1)} = D$ and $\ovf^{(2)} = -D$.

Figure~\ref{fig:exSpatial:r} plots the spatial variational terms resulting from the true basis functionals. 
Here, the reaction rates used are as before $k_d = 3 k_b/2 = 1.5$, with a diffusion constant of $D=1$. 
Contrary to the well-mixed case, these terms do not resemble step functions, but rather exhibit some extended temporal dynamics. 

\begin{figure}[!ht]
	\centering
	\includegraphics[width=0.6\columnwidth]{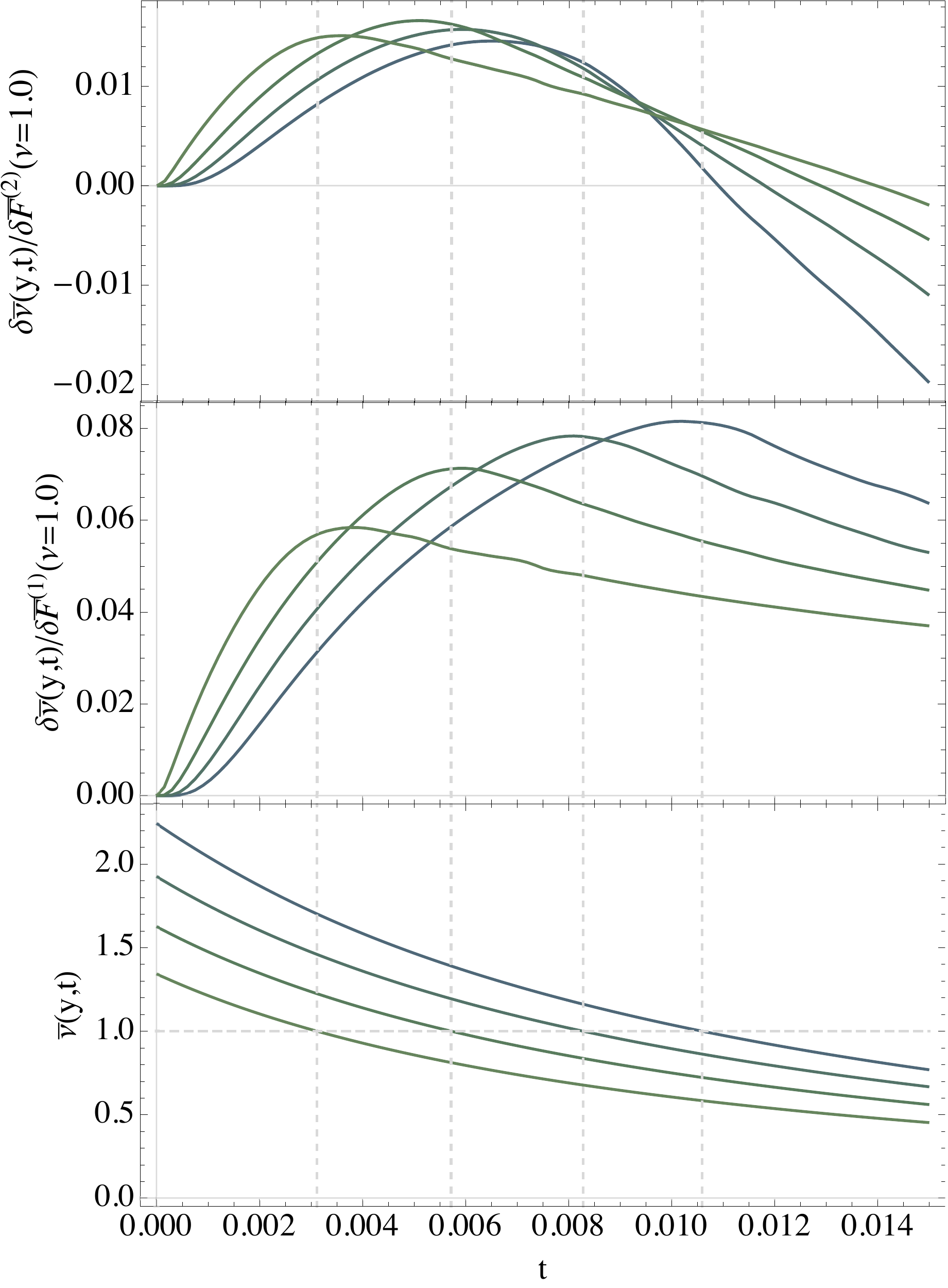}
	\caption{
    Branching random walk with diffusion in 1D estimated by Algorithm~1.
    \textit{Top:} The variational term $\delta \ovnu(y,t) / \delta \ovf^{(2)}(\nu=1.0)$ as a function of time at several spatial locations $y=0.25, 0.5, 0.75, 1$. Here, $\ovf^{(1)} = D,\ovf^{(2)}=-D$ are the true solutions. \textit{Middle:} $\delta \ovnu(y,t) / \delta \ovf^{(1)}(\nu=1.0)$. \textit{Bottom:} The solution trajectories $\ovnu(y,t)$, starting from a point source. Contrary to the well-mixed case, the variational terms do not resemble step functions at $\nu=1.0$, but rather exhibit some extended temporal dynamics.
    }
	\label{fig:exSpatial:r}
\end{figure}


\section{\label{sec:3}Estimating Effective Reduced Dynamics in 1D}

The PDE-constrained optimization problems above are the general solution for finding the basis functionals that govern the time evolution of the reduced MaxEnt model. 
Here, we present a more efficient machine learning approach for learning the basis functions from the solutions of simple, analytically solvable models. 
In Section~\ref{sec:3.A}, we present a method for finding such analytic solutions in the discrete lattice limit, and present examples for a variety of simple processes in Section~\ref{sec:3.B}. 
In Section~\ref{sec:3.C}, we demonstrate the utility of using such analytic solutions in a Boltzmann machine-like learning algorithm, and further in Section~\ref{sec:3.D} to learn non-linear combinations of solutions using artificial neural networks (ANNs).


\subsection{\label{sec:3.A}Mapping to Spin Glass Systems in 1D}

At low particle densities, a feasible model of a reaction-diffusion system in one spatial dimension and one species is that of a 1D lattice in the single occupancy limit. 
Let the spin values occupying each lattice site be $s_i \in \{0,1\}$, for all $i = 1, \dots, N$, denoting the absence or presence of a particle.

The reduced model~(\ref{eq:boltz}) now becomes the discrete analogue. 
We note that this model is consistent with the continuous version in some parameter regime where the separation between molecules is large compared to the interaction radius. 
By including only self-interactions described by an interaction function $h(t)$, and two particle nearest-neighbor interactions $J(t)$, we obtain the well known Ising model, with partition function:
\begin{equation}
\mathcal{Z} = \sum_{\{ s \}} \exp[h(t) \sum_{i=1}^N s_i + J(t) \sum_{i=1}^{N-1} s_i s_{i+1}] .
\label{eq:z}
\end{equation}

This may be evaluated explicitly using the standard transfer matrix method. 
In the thermodynamic limit, $\ln \mathcal{Z} \approx \lambda_+^N$ is analytically accessible, where $\lambda_+$ is the largest eigenvalue of the transfer matrix.

The inverse Ising problem has the solution:
\begin{equation}
\begin{pmatrix}
\left \langle \sum_{i=1}^N s_i \right \rangle (t) \\
\left \langle \sum_{i=1}^{N-1} s_i s_{i+1} \right \rangle (t)
\end{pmatrix}
=
\begin{pmatrix}
\partial_h \ln \mathcal{Z} \\
\partial_J \ln \mathcal{Z}
\end{pmatrix} .
\label{eq:invIsing}
\end{equation}
Taking the derivatives of both sides of~(\ref{eq:invIsing})
\begin{equation}
\begin{pmatrix}
\frac{d}{dt} \left \langle \sum_{i=1}^N s_i \right \rangle \\
\frac{d}{dt} \left \langle \sum_{i=1}^{N-1} s_i s_{i+1} \right \rangle
\end{pmatrix}
= 
\begin{pmatrix}
\partial_h^2 \ln \mathcal{Z} & \partial_h \partial_J \ln \mathcal{Z} \\
\partial_h \partial_J \ln \mathcal{Z} & \partial_J^2 \ln \mathcal{Z}
\end{pmatrix}
\begin{pmatrix}
\frac{dh}{dt} \\
\frac{dJ}{dt}
\end{pmatrix} .
\label{eq:nonInverted}
\end{equation}

The time derivatives of the moments on the left may be obtained directly from the CME $\dot{p} = \matr{W} p$ using the Doi-Peliti formalism described in Section~\ref{sec:2.A}. 
If the system is \textit{linear}, these may be expressed further in terms of $h,J$ using~(\ref{eq:invIsing}), and the basis functions are given directly by inverting~(\ref{eq:nonInverted}). 
If the system is \textit{non-linear}, the presence of a moment hierarchy requires an approximation in the form of a moment closure technique. 
Here, we choose to express the higher order moments that appear through the CME in terms of $h,J$, which is possible for any higher order moment since the partition function~(\ref{eq:z}) is analytically accessible. 
As a result of inverting~(\ref{eq:nonInverted}):
\begin{equation}
\begin{split}
\begin{pmatrix}
\tilde{F}_h (h,J) \\
\tilde{F}_J (h,J) 
\end{pmatrix}
& =
\begin{pmatrix}
\partial_h^2 \ln \mathcal{Z} & \partial_h \partial_J \ln \mathcal{Z} \\
\partial_h \partial_J \ln \mathcal{Z} & \partial_J^2 \ln \mathcal{Z}
\end{pmatrix}^{-1}
\times 
\begin{pmatrix}
\frac{d}{dt} \left \langle \sum_{i=1}^N s_i \right \rangle \\
\frac{d}{dt} \left \langle \sum_{i=1}^{N-1} s_i s_{i+1} \right \rangle
\end{pmatrix} ,
\label{eq:analyticBasis}
\end{split}
\end{equation}
where the RHS has been expressed in terms of $h,J$ as described above, and we use the notation $\tilde{F}_h,\tilde{F}_J$ to indicate that these are generally only approximations to the true basis functions $F_h,F_J$, and only exact for systems with closed moments.
Effectively, we have replaced the probability distribution $p$ in the CME $\dot{p}=W p$ by the dynamic Boltzmann distribution $\pt$, and evaluated the effect of the operator on the RHS on this new distribution. 
The analytic solution to the 1D inverse Ising problem therefore provides an elegant approach to moment closure (see Ref.~\onlinecite{johnson_2015},\onlinecite{smadbeck_2013} for related MaxEnt approaches to moment closure).
Similar extensions to 2D Ising models~\cite{onsager_1944} are likewise possible, and possibly to 3D as well~\cite{elshowk_2012}. 

Furthermore, we note that analogous to the continuous case proven in Proposition~1, the linearity of reaction operators in the CME extends to the basis function approximations $\tilde{F}$ (regardless of whether $\mathcal{Z}$ is analytically accessible as in the 1D case). 
This requires that the inverse Ising problem has not changed, as discussed further in Section~\ref{sec:3.C}.


\subsection{\label{sec:3.B}Analytic Approximations to Basis Functions of Simple Reaction Motifs}

Figure~\ref{fig:basisfunctions} shows the basis function approximations calculated using the 1D Ising model~(\ref{eq:analyticBasis}) for several simple unimolecular reaction processes. 
Note that the reaction rates/diffusion constant provide an overall multiplicative factor to each process. 
Computer algebra systems can be used to determine these analytic forms, which contain sums on the order of ten to a hundred terms in length, depending on the operator (see Supplementary Information for the code used to generate Figures~\ref{fig:basisfunctions},\ref{fig:trivalent}).


\begin{figure}[!ht]
	\centering
	\includegraphics[width=0.8\columnwidth]{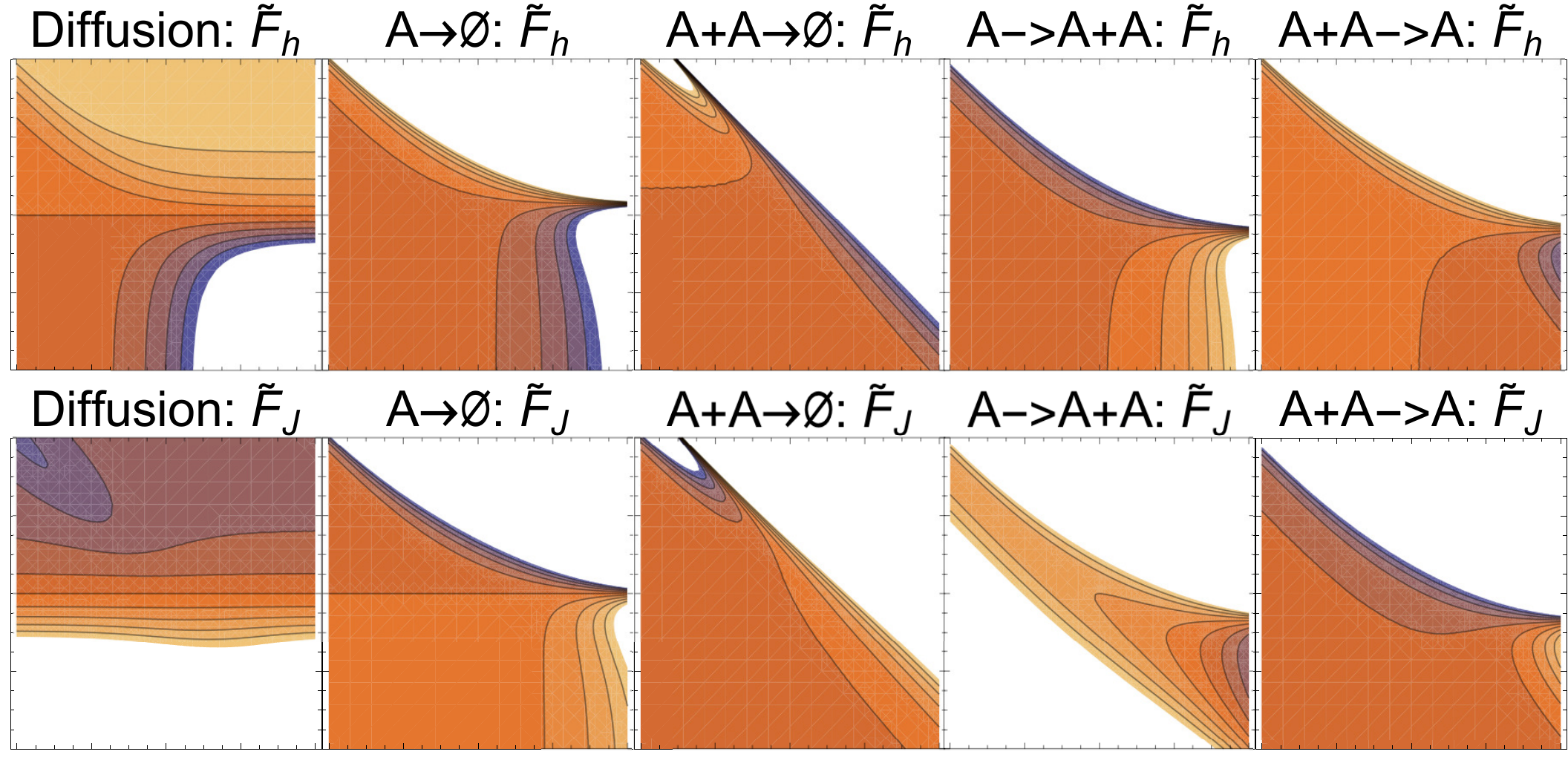}
	\caption{Basis functions~(\ref{eq:analyticBasis}) for several simple reaction schemes in one species. Horizontal, vertical axis: $h,J \in [-4,4]$. The magnitudes have been scaled to $[-1,1]$, since the reaction rate/diffusion constant provides an arbitrary scaling factor.}
	\label{fig:basisfunctions}
\end{figure}

Generalizing these simple systems, we solve for the basis function approximations of the trivalent reaction $A + B \rightarrow C$ with its reverse process $C \rightarrow A + B$. 
This process is fundamentally important as a generalization of many simple biochemical processes, and has been studied extensively~\cite{mjolsness_2013_2,mjolsness_2013_1}. 
For example, it is the building block of the broadly applicable substrate-enzyme-product (SEP) motif $S + E \rightleftharpoons C \rightarrow P + E$, where $S,E,P$ denote the substrate, enzyme, and product (see Section~\ref{sec:3.D} below). 

In the Ising model formalism, the description of this process involves 9 time dependent interaction functions $h_A,h_B,h_C,J_{AA},J_{AB},J_{AC},J_{BB},J_{BC},J_{CC}$, forming the reduced model:
\begin{equation}
\mathcal{Z} = \sum_{\{ s \}} \sum_{\{ \alpha \}} \exp [ \sum_{i=1}^N h_{\alpha_i} (t) s_i + \sum_{i=1}^{N-1} J_{\alpha_i,\alpha_{i+1}}(t) s_i s_{i+1} ]
\end{equation}
where the species label $\alpha_i \in \{ A, B, C \}$, and we implicitly note that the sum $\sum_{\{\alpha\}}$ runs only over occupied sites $s_i = 1$. 
Figure~\ref{fig:trivalent} shows several 2D slices of three of the nine basis function approximations for the forward process $A + B \rightarrow C$.

By including species labels,~(\ref{eq:analyticBasis}) leads to analytic expressions containing on the order of hundreds of terms. 
Here, we used a numerical strategy as described in Appendix~\ref{app:evalbfnum} for evaluating the basis functions over the chosen domain. 
While a computer algebra system may be employed as before, this strategy is computationally faster.


\begin{figure}[!ht]
	\centering
	\includegraphics[width=0.8\columnwidth]{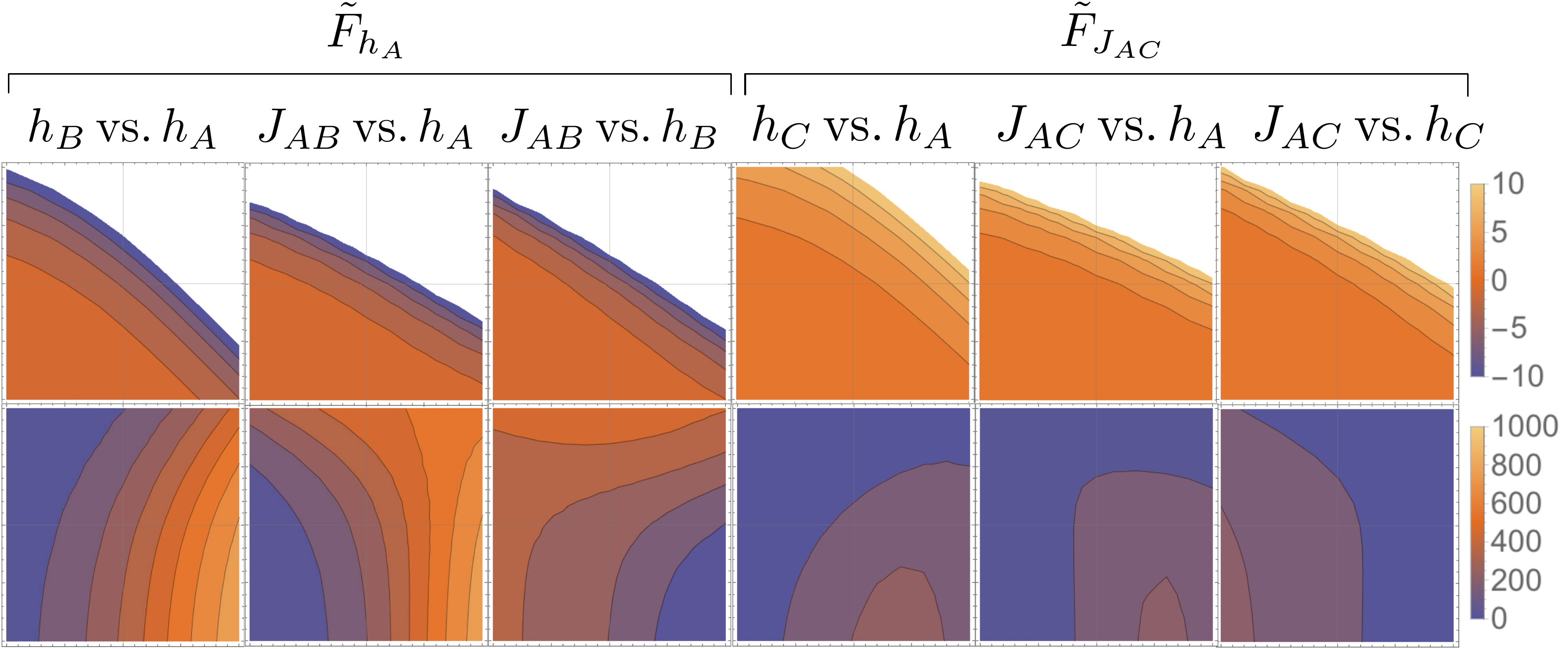}
  \caption{Basis function approximations $\tilde{F}_{h_A},\tilde{F}_{J_{AC}}$ corresponding to the forward trivalent reaction $A + B\rightarrow C$ with rate $k=1$. Each is a 9 dimensional function, of which 2D slices are shown, holding all other parameters at zero. The top row shows the basis function, while the bottom row shows the corresponding moments controlled by these parameters $h_A,J_{AC}$. The chain length used is $N=1000$. The ranges for all horizontal, vertical axes are $[-2,2]$.}
\label{fig:trivalent}
\end{figure}


\subsection{\label{sec:3.C}Boltzmann Machine-Style Learning Algorithm for Dynamics}

The basis function approximations derived above constitute a space of possible reduced dynamics. 
Here, we consider using these analytic insights to describe large spatially distributed reaction networks in 1D. This approach faces two key problems:
\begin{enumerate}
\item For non-linear systems, $\pt$ obeying~(\ref{eq:analyticBasis}) will over time diverge from the MaxEnt distribution consistent with the CME moments due the moment closure approximation made. 
As a fundamental consequence of this moment hierarchy, it is not possible to find exact basis functions over the \textit{entire} interaction parameter space (e.g. $h,J$). 
Another way to see this is that trajectories of the CME system will intersect in $h,J$ space.

However, we postulate that it may be possible to learn approximately well the basis functions for a \textit{single trajectory} (from a single initial condition) which does not self-intersect over some domain. 
This model may be used for extrapolation with reasonable accuracy close to the stochastic trajectory.

\item For large reaction networks, the basis functions are generally not linear in the basis functions of individual processes because the collection of interaction functions is not fixed, violating the assumption in Proposition~1. 
For example, consider the process $A \rightarrow B \rightarrow C$. 
Here, nine basis functions are required to capture all means and nearest neighbor correlations, such that~(\ref{eq:analyticBasis}) is nine dimensional. Denote these by $\bb = {\matr{A}}^{-1} \mb$ where $\bb$ denotes the basis functions, $\mb$ the time evolving moments, and ${\matr{A}}$ the matrix of partition function derivatives.

Next, consider the separate processes $A \rightarrow B$ and $B \rightarrow C$, described by five basis functions each. 
Let these be denoted by $\bb^{(r)} = ({\matr{A}}^{(r)})^{-1} \mb^{(r)}$ for each of the two reactions $r$. 
Clearly, not all nine basis functions in $\bb$ are present in each $\bb^{(r)}$. 
Furthermore, for those that are present in both, it is not necessarily true that the $i$-th basis function is expressible as $\bb_i \neq \bb_j^{(1)} + \bb_k^{(2)}$ for appropriate $j,k$.

Generally, a reaction network involves more interaction parameters than each of the individual processes, such that Proposition~1 does not 
apply. 
It is only for a subset of networks, such as reaction networks in one species, where the linearity in the CME extends \textit{exactly} to the basis functions. 
Regardless, we postulate that many networks may be described \textit{approximately} well by linear combinations of basis functions corresponding to individual processes.
\end{enumerate}

In light of these postulates, we return to the variational problem~(\ref{eq:varProblem}) and its PDE-constraint. 
In the discrete lattice case considered in Section~\ref{sec:3.A}, it becomes for each $\gamma = h,J$:
\begin{equation}
\int_0^\infty d\tp \; \left ( \tilde{\mu}(\tp) - \mu(\tp) \right ) \frac{\delta h(\tp) }{\delta F_\gamma(h,J) } + \int_0^\infty d\tp \; \left ( \tilde{\Delta}(\tp) - \Delta(\tp) \right ) \frac{\delta J(\tp) }{\delta F_\gamma(h,J) } = 0 ,
\label{eq:bmFull}
\end{equation}
where we have used the notation $\mu,\Delta$ to denote the average number of particles, nearest neighbors (NN) over $p$, and similarly $\tilde{\mu},\tilde{\Delta}$ to denote averages over $\pt$.

Here, we exploit the analytic results derived above to simplify this problem and derive an efficient Boltzmann-machine type learning algorithm for the dynamics. 
In particular, we assume that the true basis functions are linear combinations of the approximations derived in Section~\ref{sec:3.B} above, given by:
\begin{equation}
\begin{split}
\frac{dh}{dt} &= F_h (h,J) = \sum_r \theta^{(r)} \tilde{F}_h^{(r)}, \\
\frac{dJ}{dt} &= F_J (h,J) = \sum_r \theta^{(r)} \tilde{F}_J^{(r)}.
\end{split}
\label{eq:bmCons}
\end{equation}
Here, the reaction rates and diffusion constant are all set to unity, such that the coefficients $\theta$ indicate the rates. 
The variational problem now turns into a regular optimization problem for the coefficients $\theta$ that will yield at all times the MaxEnt distribution consistent with the CME moments. 
The optimization problem becomes: Subject to the PDE constraint~(\ref{eq:bmCons}), solve:
\begin{equation}
\int_0^\infty d\tp \; \left ( \tilde{\mu}(\tp) - \mu(\tp) \right ) \frac{\partial h(\tp)}{\partial \theta^{(s)} } + \int_0^\infty d\tp \; \left ( \tilde{\Delta}(\tp) - \Delta(\tp) \right ) \frac{\partial J(\tp)}{\partial \theta^{(s)} } = 0 ,
\label{eq:bmObjective}
\end{equation}
where the derivative terms are given by the solution to the ordinary differential equation system
\begin{widetext}
\begin{equation}
\begin{split}
\frac{\partial}{\partial \tp} \left ( \frac{\partial h(\tp)}{\partial \theta^{(s)} } \right )
=&
\tilde{F}_h^{(s)} + \frac{\partial h(\tp)}{\partial \theta^{(s)}} \sum_r \theta^{(r)} \frac{\partial \tilde{F}_h^{(r)}}{\partial h} 
+ \frac{\partial J(\tp)}{\partial \theta^{(s)}} \sum_r \theta^{(r)} \frac{\partial \tilde{F}_h^{(r)}}{\partial J} ,
\\
\frac{\partial}{\partial \tp} \left ( \frac{\partial J(\tp)}{\partial \theta^{(s)} } \right )
=&
\tilde{F}_J^{(s)} + \frac{\partial h(\tp)}{\partial \theta^{(s)}} \sum_r \theta^{(r)} \frac{\partial \tilde{F}_J^{(r)}}{\partial h} 
+ \frac{\partial J(\tp)}{\partial \theta^{(s)}} \sum_r \theta^{(r)} \frac{\partial \tilde{F}_J^{(r)}}{\partial J} ,
\end{split}
\label{eq:bmVar}
\end{equation}
\end{widetext}
with initial condition $\partial h(0) / \partial \theta^{(s)} = \partial J(0) / \partial \theta^{(s)} = 0$.

Parameter estimation is greatly simpler to solve than the function estimation~(\ref{eq:bmFull}). 
Furthermore, the variational problem~(\ref{eq:bmVar}) is significantly simplified, since $\tilde{F}^{(r)}$ and consequentially its derivatives are analytically accessible. 
We capitalize upon these practical qualities in Algorithm~2, which solves this problem in a Boltzmann-machine learning style approach.

\begin{figure}[t]
\begin{algorithm}[H]
\caption{Boltzmann Machine-Style Learning of Dynamics} \label{alg:2}
\begin{algorithmic}[1]
\algsetblock[Init]{Initialize}{Stop}{4}{1.4em}
\Initialize
\State Initial $\theta^{(r)}$ for all $r$. \;
\State Max. integration time $T$. \;
\State A formula for the learning rate $\lambda$. \;
\State Time-series of lattice spins $\{ s \}(t)$ from stochastic simulations from some known IC $h_0,J_0$. \;
\State\hspace{\algorithmicindent}Fully visible MRF with NN connections and as many units as lattice sites $N$. \;
\While{not converged}
	\LineComment{\textit{Generate trajectory in reduced space:}}
	\State Solve the PDE constraint~(\ref{eq:bmCons}) with IC $h_0,J_0$ for $0 \leq t \leq T$. \;
	\LineComment{\textit{Awake phase:}}
    \State Evaluate true moments $\mu(t),\Delta(t)$ from the stochastic simulation data $\{s\}(t)$.  \;
    \LineComment{\textit{Asleep phase:}}
	\State Evaluate moments $\tilde{\mu}(t),\tilde{\Delta}(t)$ of the Boltzmann distribution by Gibbs sampling. \;
	\LineComment{\textit{Update to decrease objective function:}}
	\State Solve~(\ref{eq:bmVar}) for derivative terms. \;
	\State Update $\theta^{(s)}$ to decrease the objective function for all $s$ by taking: $\theta^{(s)} \rightarrow \theta^{(s)} - \lambda \times (\ref{eq:bmObjective})$. \;
\EndWhile
\end{algorithmic}
\end{algorithm}
\end{figure}

\begin{figure*}[!ht]
	\centering
	\includegraphics[width=1.0\textwidth]{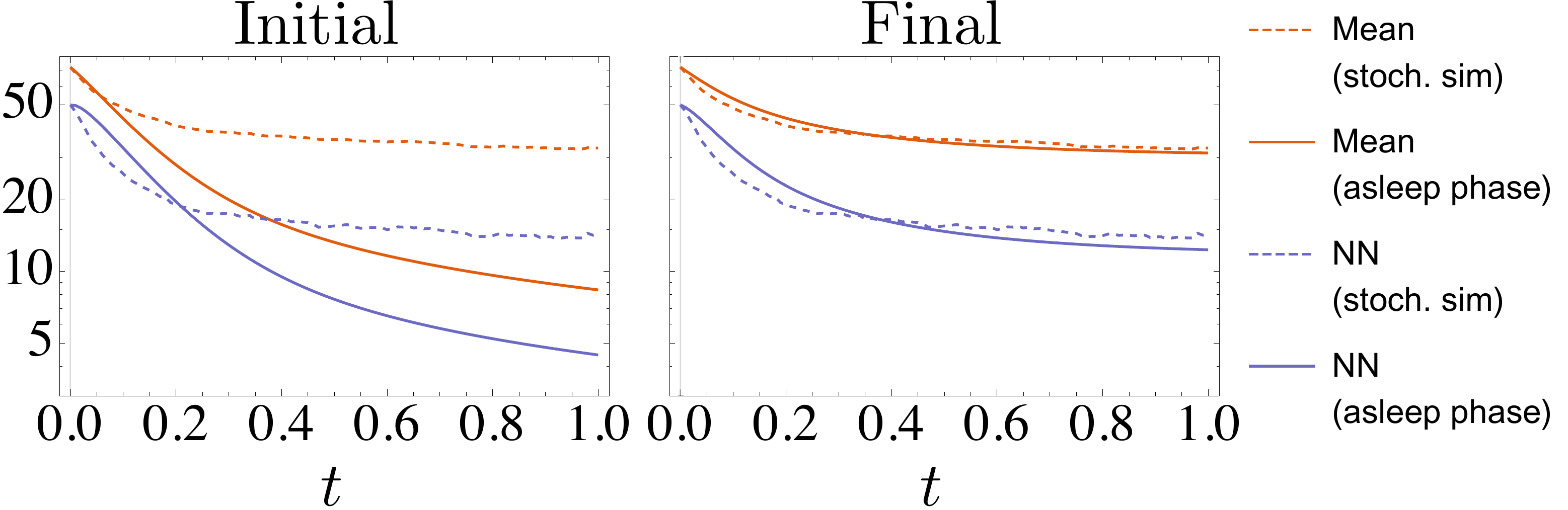}
	\caption{The 
    1$^{\text{st}}$ and 2$^{\text{nd}}$ (mean and NN) moments of the BARW system obtained from stochastic simulation (dashed) and by integrating the PDE constraint~(\ref{eq:bmCons}) and using Gibbs sampling in the asleep phase of Algorithm~2 (solid). \textit{Left:} using initial $\theta_0^{(s)}$ reveals the limitations of moment closure approximation. \textit{Right:} after 400 iterations, the coefficients have adjusted to more accurately capture the true CME dynamics.}
	\label{fig:barwMoments}
\end{figure*}

\begin{figure}[!ht]
	\centering
	\includegraphics[width=0.5\columnwidth]{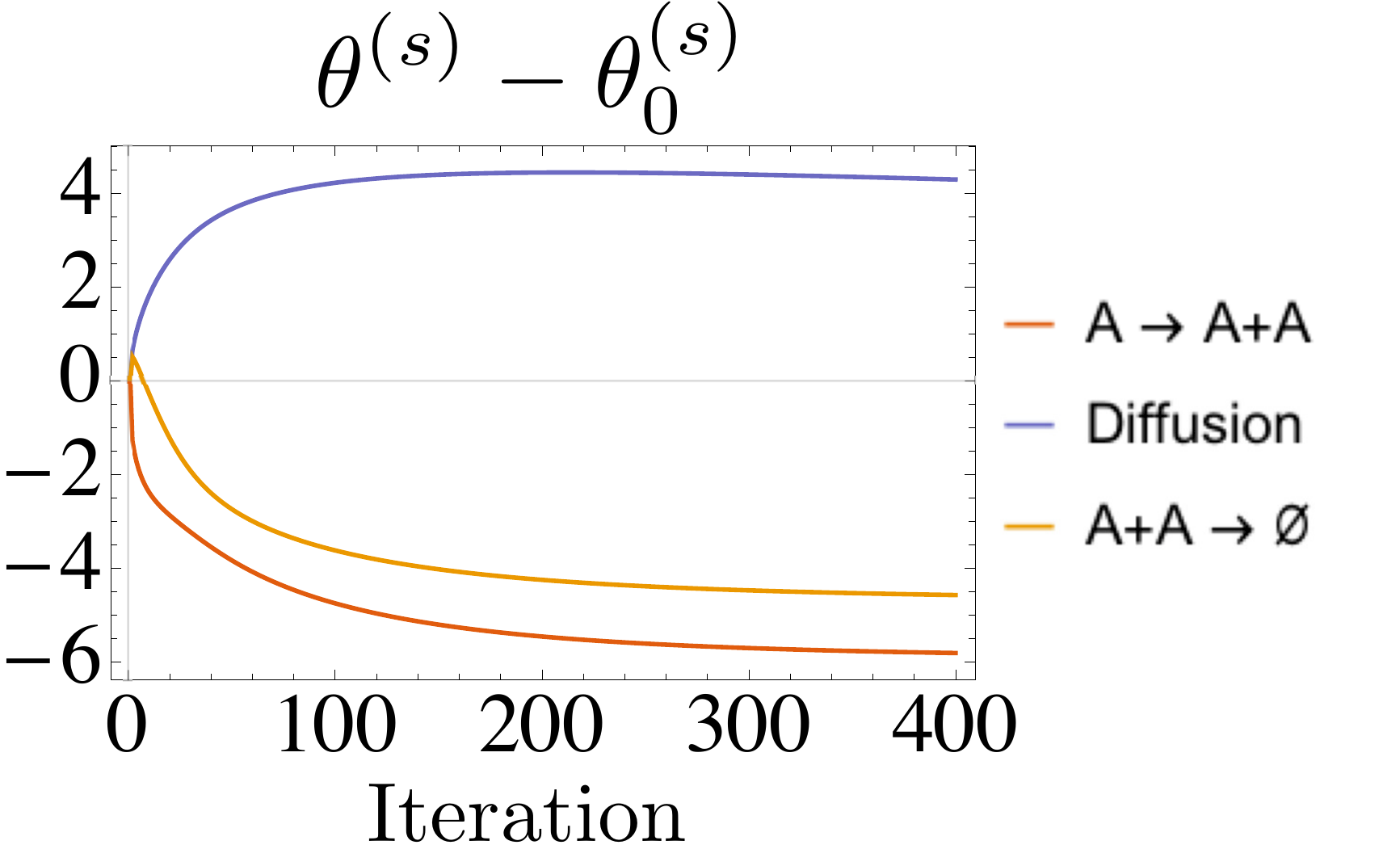}
	\caption{The coefficients in the PDE constraint converging over 400 iterations of Algorithm~2 applied to the BARW system, starting from the parameters used in the stochastic simulations.}
	\label{fig:barwTheta}
\end{figure}

As an illustrative example, we apply Algorithm~2 to a branching and annihilating random walk (BARW) on a 1D lattice, described by the three processes: $A\rightarrow A+A$ with rate $k_b = 10$, $A+A\rightarrow 0$ with rate $k_a = 10$, and diffusion with constant $D=10$. 
Extensive theoretical work has been dedicated to studying BARWs in the context of universality classes, in particular the directed percolation universality class~\cite{takayasu_1992,cardy_1996}.

Stochastic simulations are used to generate training data for this system on a chain of length $N=100$ for maximum time of $T=1$ with timestep $dt=0.01$. 
Here, we follow the numerical procedure described in Ref.~\onlinecite{takayasu_1992}. 
The basis functions used in~(\ref{eq:bmCons}) are those of the three processes present, as shown in Figure~\ref{fig:basisfunctions}. 
The initial coefficients $\theta_0^{(s)}$ used are the known reaction rates. 

Figure~\ref{fig:barwMoments} shows the moments of the BARW system. 
Due to the moment closure problem, the system predicted by solving the constraint equations diverges from the true, even though the true reaction rates are used as coefficients $\theta_0^{(s)}$. 
After running 400 iterations of Algorithm~2, the new coefficients lead to much closer agreement to the true system.

Figure~\ref{fig:barwTheta} shows the coefficients converge over the iterations. 
In particular, the effective rates for bimolecular annihilation and branching have decreased, while the effective diffusion constant has increased. 
Since the final values are sensitive to the initial $\theta_0^{(s)}$ chosen, an $L_2$ regularization term is included in the action. 
A further constraint in Algorithm~2 to keep $\theta^{(s)}$ positive enforces the connection to effective reaction rates.


\subsection{\label{sec:3.D}ANNs for Learning Non-linear Combinations of Basis Functions}

As a more general approach than linear combinations, we use ANNs (artificial neural networks) to describe non-linear combinations of basis functions. 
Consider the SEP system diffusing on a 1D lattice, described by the reactions:
\begin{equation}
S + E \xrightleftharpoons[k_{-1}]{k_1} C \xrightarrow{k_2} P + E .
\end{equation}
The full Ising model for this system consists of four self interactions and 10 NN coupling parameters.

Figure~\ref{fig:sep} shows several moments of this system evolving in time from stochastic simulations. 
Here, the parameters used are: $k_1 = 10,k_{-1}=0.1,k_2=0.5$, max. time $T=1$ with timestep $0.01$, and lattice length $N=100$. 
The system evolves from an initial lattice generated by Gibbs sampling with parameters $h_S=0.5,h_E=1,h_C=-1,h_P=-1$, and all NN terms set to zero.

The input to the ANN are the basis functions for the three separate processes, each of which belongs to the trivalent reaction motif of Figure~\ref{fig:trivalent} thereby contributing 9 basis functions. 
Additionally, the two basis functions for the diffusion of each of the four species is included from Figure~\ref{fig:basisfunctions}, for a total of 35 inputs. 
The other layers in the ANN are two layers of 40 units, and an output layer of 14 units, with $\tanh$ activation functions between each layer. Two thirds of the total length $T$ of the timeseries are used for training. 
These are converted to trajectories in interaction parameter space using Boltzmann machine learning, and smoothed using a low-pass filter before being used to evaluate the 35 input basis functions. 
The corresponding outputs to be learned are the time derivatives of these 14 parameters, also smoothed by a low-pass filter. 

The network learns the dynamics of these parameters to high precision. 
We infer from the fast training times that the usage of these analytic solutions as input greatly reduces the difficulty of training the network from the interaction parameters directly.

Figure~\ref{fig:sep} shows the extrapolated parameters and corresponding moments, compared to the remaining third of the simulation time. 
These extrapolations are generally linear in interaction space, and may diverge quickly, such as for $h_S$. 
However, the moments show considerable robustness to these variations, suggesting that using ANNs for extrapolation is possible. 
This has promising implications for further development in multiscale simulation algorithms.

A further feature learned by the ANN is a moment closure approximation for the dynamics of $J_{SP}$, and the corresponding NN moment it controls. 
This parameter is not included in any of the basis functions or inputs to the ANN. 
The basis function learned, shown in Figure~\ref{fig:sep}, therefore expresses the dynamics of this moment in terms of the interactions made available as input to the network. 
Similar extensions to higher order moments are likewise possible.

\begin{figure*}[!ht]
	\centering
	\includegraphics[width=1.0\textwidth]{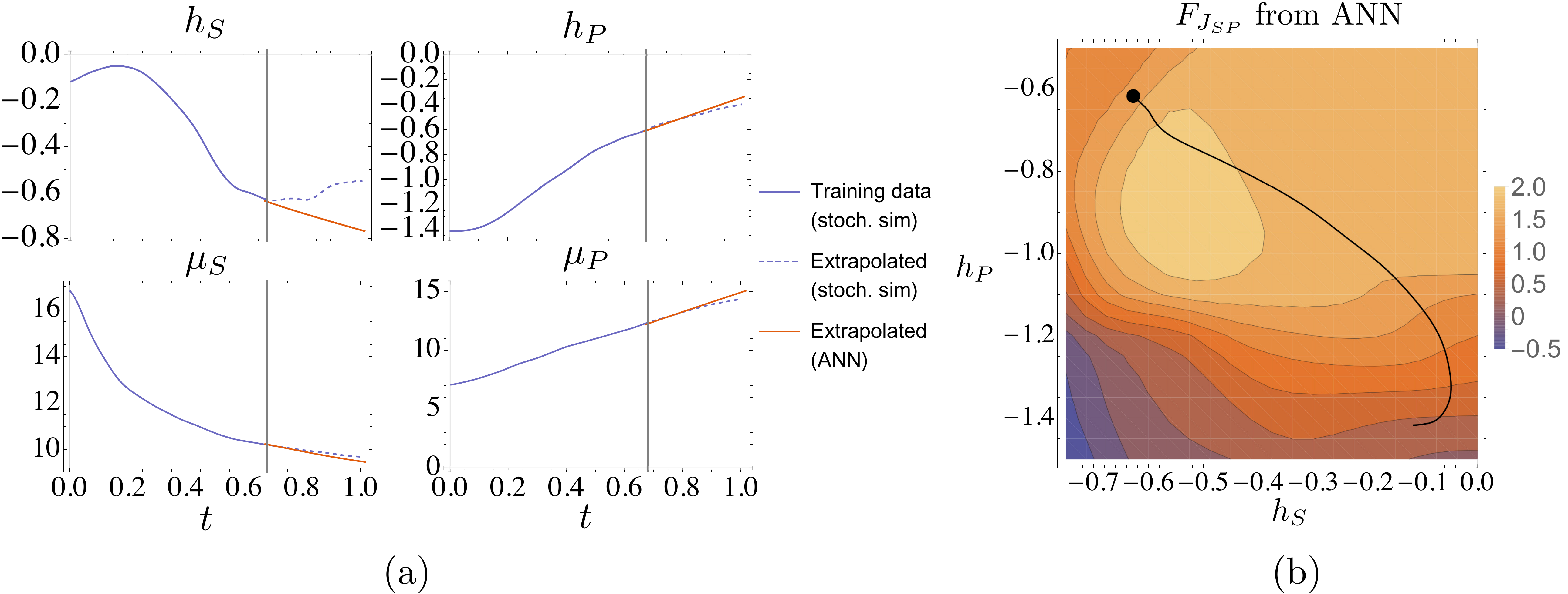}
	\caption{(a) Trajectories of the SEP system from stochastic simulation, and extrapolated values from the trained ANN. The divergence of the predicted and true values in moment space is smaller than in interaction parameter space, suggesting a stability in the observable quantities of the model to small errors. (b) The derivative learned by the ANN for the moment. A two dimensional slice is shown through this 14 dimensional function. The black line shows the trajectory of the training data, while the dot indicates the evaluation point for this slice, chosen at the end of the training data (gray vertical line in (a)). All other parameters than $h_S,h_P$ are held fixed at this point.}
	\label{fig:sep}
\end{figure*}


\section{\label{sec:4}Discussion and Conclusions}

This paper presents a new approach to model reduction of spatial chemical systems. 
Slowly time-evolving MaxEnt models are employed to capture the key correlations in the system. 
This approach is particularly useful for multiscale problems, where different spatial and temporal correlations become more or less relevant over time to accurately describe the system. 
For example, in synaptic level neuroscience, the stochastic influx of signaling molecules in the post-synaptic spine produces complex spatial correlations between ion channels and downstream targets, but these are less relevant during quiescent periods. 
We anticipate that such problems stand to benefit greatly from modeling approaches that are able to adjust which correlations are included to optimize simulation efficiency and accuracy.

A general model that is functional in nature is introduced to describe dynamic Boltzmann distributions. 
This extends and formalizes ideas originally developed in GCCD in Ref.~\onlinecite{johnson_2015} - in particular:
\begin{enumerate}
\item A general variational problem has been formulated to determine the functions in the dynamical system controlling the interaction parameters. This takes the form of a PDE-constrained optimization problem.
\item The reduced model has been extended to capture spatial correlations, with particular relevance to Biological applications. 
By motivating parameterizations of the functionals from analytically solvable cases, practical optimization algorithms for learning the dynamics of spatial systems are made possible.
\item ANNs have been employed to learn \textit{non-linear} combinations of basis functions, derived for individual reaction processes using the aid of computer algebra systems. 
\end{enumerate}

Mapping the chemical system onto a spin lattice allows a direct connection to the more traditionally formulation of a Boltzmann machine. 
Here, the connection to the new learning algorithm is evident in~(\ref{eq:bmFull}), and we anticipate this will suggest numerous further applications to diverse areas of machine learning where estimating the dynamics of a time series is required. 
Including arbitrary spatial correlations beyond NN in the lattice model may be of further interest in pursuit of 3D simulations.

Numerous strategies are possible for improving the efficiency of the PDE-constrained optimization problem formulated here, such as adjoint methods~\cite{giles_2000}.
In this work, we have shown that the complexity of this problem can be greatly reduced by instead learning linear and non-linear combinations of analytically accessible approximations. 
Deconstructing the problem in this way can offer physical insight into a complex reaction system, such as in Section~\ref{sec:3.C} where effective reaction rates are learned. 
Future work in this direction may further explore these principled methods for integrating human intuition with machine inference in the model reduction process.


\section*{Supplementary Material}

See supplementary material for alternate derivations of the differential equation system~(\ref{eq:pdeWM}), and for code used to implement algorithms~1 and~2.


\begin{acknowledgments}

This work was supported by NIH grants R01HD073179 and USAF/DARPA FA8750-14-C-0011 (E.M.) and NIH P41-GM103712 and AFOSR MURI FA9550-18-1-0051 (O.K.E., T.B., T.S.).

\end{acknowledgments}


\appendix



\section{\label{app:varTerm}Derivation of Differential Equation System for Variational Term}



\subsection{\label{app:varTerm:WM}Well-Mixed Case}

Consider the differential equation system~(\ref{eq:adeWM}). Represent the solution as a functional of the basis functions $F$ using the notation
\begin{equation}
\nu_\kp(\tp) = J_\kp [ \{ F \} ] ,
\end{equation}
where $\{ F \} = \{ F_l \; | \; l = 1, \dots, K\}$,
and $J$ results from solving (\ref{eq:adeWM}). 
Further, let $\{ J[ \{ F \} ] \} = \{ J_l [ \{ F \} ] \; | \; l=1,\dots,K \}$, then~(\ref{eq:adeWM}) is:
\begin{equation}
\frac{d}{d\tp} J_\kp [ \{ F \} ] = F_\kp ( \{ J[ \{ F \} ] \} ) .
\end{equation}
To find the variational term $\delta \nu_\kp(\tp) / \delta F_k (\{ \nu \} )$, let $F_k \rightarrow F_k + \epsilon \eta$ using the notation 
\begin{equation}
\{ F^\prime \} = \{ F_l + \delta_{l,k} \epsilon \eta | l = 1, \dots, K \} ,
\end{equation}
then:
\begin{equation}
\begin{split}
\frac{d}{d\tp} J_\kp [ \{ F^\prime \} ] 
= F_\kp ( \{ J[ \{ F^\prime \} ] \} )
+ \delta_{\kp,k} \epsilon \eta ( \{ J[ \{ F^\prime \} ] \} ) .
\end{split}
\end{equation}
Differentiating with respect to $\epsilon$ at $\epsilon=0$ gives:
\begin{equation}
\frac{d}{d\tp} \left ( \frac{d J_\kp [ \{ F^\prime \} ]}{d \epsilon} \Bigg |_{\epsilon=0} \right )
=
\sum_{l=1}^K
\frac{\partial F_\kp(\{ \nu(\tp) \})}{\partial \nu_l(\tp)}
\left ( \frac{d J_l [ \{ F^\prime \} ]}{d \epsilon} \Bigg |_{\epsilon=0} \right )
+ \delta_{\kp,k} \eta ( \{ \nu(\tp) \} ) .
\end{equation}
Substitute the definition of the functional derivative
\begin{equation}
\frac{d J_\kp [ \{ F^\prime \} ]}{d \epsilon} \Bigg |_{\epsilon=0} = \int d \nu_1 \dots \int d \nu_K \; \frac{\delta \nu_\kp(\tp)}{\delta F_k (\{ \nu \})} \eta(\{ \nu \})
\end{equation}
to obtain~(\ref{eq:pdeWM}):
\begin{equation}
\begin{split}
\frac{d}{d\tp} \left ( \frac{\delta \nu_\kp(\tp)}{\delta F_k (\{ \nu \})} \right )
=
\sum_{l=1}^K
\frac{\partial F_\kp(\{ \nu(\tp) \})}{\partial \nu_l(\tp)}
\frac{\delta \nu_l(\tp)}{\delta F_k (\{ \nu \})}
+ \delta_{\kp,k} \delta ( \{ \nu \} - \{ \nu(\tp) \} ) .
\end{split}
\end{equation}


\subsection{\label{app:varTerm:diff}Spatially Heterogeneous Example: Diffusion in 1D}

Consider a diffusion process in 1D, described by single basis functional parameterized according to:
\begin{equation}
\frac{d}{d\tp} \nu(\yp,\tp) = F[\nu (\yp, \tp) ]
F^{(1)} (\nu(\yp,\tp)) \left ( \partial_{\yp} \nu(\yp,\tp) \right )^2 
+ F^{(2)} (\nu(\yp,\tp)) \left ( \partial_{\yp}^2 \nu(\yp,\tp) \right ) .
\end{equation}
Use the functional notation:
\begin{equation}
\nu (\yp,\tp) = J [ \{ F \} ] ,
\end{equation}
where $\{ F \} = \{ F^{(1)},F^{(2)} \}$
and $J$ results from solving (\ref{eq:adeLocal}), then:
\begin{equation}
\frac{d}{dt} J [\{ F \}]
=
F^{(1)} ( J [\{ F \}] ) \left ( \partial_{\yp} J [\{ F \}] \right )^2
+ F^{(2)} ( J [\{ F \}] ) \partial_{\yp}^2 J [\{ F \}] .
\end{equation}

To find the variational term $\delta \nu(\yp,\tp) / \delta F^{(\gamma)} (\omega)$ for $\gamma = 1,2$, let $F^{(\gamma)} \rightarrow F^{(\gamma)} + \epsilon \eta$. Use the notation:
\begin{equation}
\{ F^\prime \} = \{ F^{(1)} + \delta_{\gamma,1} \epsilon \eta, F^{(2)} + \delta_{\gamma,2} \epsilon \eta \} ,
\end{equation}
then
\begin{equation}
\begin{split}
\frac{d}{dt} J [\{ F^\prime \}]
=&
F^{(1)} ( J [\{ F^\prime \}] ) \left ( \partial_{\yp} J [\{ F^\prime \}] \right )^2
+ F^{(2)} ( J [\{ F^\prime \}] ) \partial_{\yp}^2 J [\{ F^\prime \}]
\\ & \hspace{0mm}
+ \delta_{\gamma,1} \epsilon \eta ( J [\{ F^\prime \}] ) \left ( \partial_{\yp} J [\{ F^\prime \}] \right )^2
+ \delta_{\gamma,2} \epsilon \eta ( J [\{ F^\prime \}] ) \partial_{\yp}^2 J [\{ F^\prime \}] .
\end{split}
\end{equation}
Take the derivative with respect to $\epsilon$ at $\epsilon=0$:
\begin{equation}
\begin{split}
\frac{d}{dt} \left ( \frac{d J [\{ F^\prime \}]}{d \epsilon} \Bigg |_{\epsilon=0} \right )
&=
\left (
\frac{\partial F^{(1)} (\nu) }{ \partial \nu }
\left ( \partial_{\yp} \nu \right )^2
+
\frac{\partial F^{(2)} (\nu) }{ \partial \nu }
\partial_{\yp}^2 \nu
\right )
\left ( \frac{d J [\{ F^\prime \}]}{d \epsilon} \Bigg |_{\epsilon=0} \right )
+
\Bigg ( \delta_{\gamma,1} \left ( \partial_{\yp} \nu \right )^2
+
\delta_{\gamma,2} \partial_{\yp}^2 \nu 
\Bigg ) 
\eta ( \nu )
\\ & \hspace{0mm}
+ 2 F^{(1)} (\nu) \partial_{\yp} \nu \frac{\partial}{\partial \yp}
\left ( \frac{d J [\{ F^\prime \}]}{d \epsilon} \Bigg |_{\epsilon=0} \right )
+ F^{(2)} (\nu) \frac{\partial^2}{\partial {\yp}^2}
\left ( \frac{d J [\{ F^\prime \}]}{d \epsilon} \Bigg |_{\epsilon=0} \right ) ,
\end{split}
\end{equation}
where $\nu = \nu(\yp,\tp)$ everywhere. Substituting the definition of the functional derivative
\begin{equation}
\frac{d J [\{ F^\prime \}]}{d \epsilon} \Bigg |_{\epsilon=0} = \int d \omega \; \frac{\delta \nu(\yp,\tp) }{\delta F^{(\gamma)}(\omega) } \eta(\omega)
\end{equation}
gives
\begin{equation}
\begin{split}
\frac{d}{dt} \left ( \frac{\delta \nu }{\delta F^{(\gamma)}(\omega) } \right )
=&
\left (
\frac{\partial F^{(1)} (\nu) }{ \partial \nu }
\left ( \partial_{\yp} \nu \right )^2
+
\frac{\partial F^{(2)} (\nu) }{ \partial \nu }
\partial_{\yp}^2 \nu
\right )
\left ( \frac{\delta \nu }{\delta F^{(\gamma)}(\omega) } \right )
+
\Bigg ( 
\delta_{\gamma,1} \left ( \partial_{\yp} \nu \right )^2
+
\delta_{\gamma,2} \partial_{\yp}^2 \nu
\Bigg )
\delta ( \nu - \omega )
\\ & \hspace{0mm}
+ 2 F^{(1)} (\nu) \partial_{\yp} \nu \frac{\partial}{\partial \yp}
\left ( \frac{\delta \nu }{\delta F^{(\gamma)}(\omega) } \right )
+ F^{(2)} (\nu) \frac{\partial^2}{\partial {\yp}^2}
\left ( \frac{\delta \nu }{\delta F^{(\gamma)}(\omega) } \right ) .
\end{split}
\label{eq:app:varTermsDiff}
\end{equation}



\section{\label{app:evalbfnum}Evaluating Basis Functions Numerically}


To compute the basis functions numerically using~(\ref{eq:analyticBasis}), an efficient method is possible if the eigenvalues of the transfer matrix $M$ are singular. Let the eigenvalues be $\lambda_i$ with corresponding eigenvectors $\boldsymbol{u}_i$. Define:
\begin{equation}
p_{ij}(\alpha) = {\boldsymbol u}_i^\intercal (\partial_\alpha M) {\boldsymbol u}_j^\intercal
\end{equation}
for $\alpha=h,J$, where $\partial_\alpha M$ denotes component-wise differentiation of $M$. Also note that $p_{ij}(\alpha)=p_{ji}(\alpha)$ is symmetric. Then the derivatives of the eigenvalues are given by:~\cite{lancaster_1964}
\begin{equation}
\begin{split}
\partial_\alpha \lambda_i 
&= p_{ii}(\alpha) , \\
\partial_\alpha \partial_\beta \lambda_i 
&= {\boldsymbol u}_i^\intercal (\partial_\alpha \partial_\beta M) {\boldsymbol u}_i + 2 \sum_{j \neq i} \frac{p_{ij}(\alpha) p_{ij}(\beta) }{\lambda_i - \lambda_j} ,
\end{split}
\end{equation}
for $\beta=h,J$. The principle advantage of this approach lies in the fact that the analytic expressions for $\partial_\alpha M$ and $\partial_\alpha \partial_\beta M$ are simpler to derive than differentiating the analytic expressions for the eigenvalues $\lambda$.

It is now straightforward to numerically evaluate the components $\partial_\alpha \partial_\beta \ln \mathcal{Z}$ of~(\ref{eq:analyticBasis}) in the thermodynamic limit $\ln \mathcal{Z} \approx N \ln \lambda_+$, where $\lambda_+$ is the largest eigenvalue of the transfer matrix and $N$ the length of the chain.


\nocite{*}
\bibliography{bibliography}

\end{document}